\documentclass[10pt]{book}
\usepackage[sectionbib]{natbib}
\usepackage{array,fancyheadings,rotating,graphics}
\usepackage{hyperref}
\usepackage[plain,noend]{algorithm2e}

\usepackage{subcaption}
\usepackage{caption}
\usepackage[noend]{algpseudocode}

\usepackage{amsmath, float}
\usepackage{amssymb}
\usepackage{amsfonts}
\usepackage{multirow,color}
\usepackage{amsthm,graphicx}

\newtheorem{theorem}{Theorem}
\newtheorem{lemma}{Lemma}

\theoremstyle{definition}
\newtheorem{definition}{Definition}
\newtheorem{example}{Example}
\newtheorem{remark}{Remark}

\newcommand{\pkg}[1]{$\mathsf{ #1}$}

\def\cA{{\cal A}}

\def\cO{{\cal O}}

\def\blambda{{\boldsymbol \lambda }}

\def\argmin{\mathop{\rm \arg\min}}
\def\real{\mathop{{\rm I}\kern-.2em\hbox{\rm R}}\nolimits}
\def\diag{\mbox{diag}}

\def\bbeta{\boldsymbol \beta}

\def\lam{\lambda}

\def\T{{ \mathrm{\scriptscriptstyle T} }}
\newtheorem{condition}{Condition}

\begin{document}


\renewcommand{\baselinestretch}{1.2}

\markright{ \hbox{\footnotesize\rm Statistica Sinica
}\hfill\\[-13pt]
\hbox{\footnotesize\rm
}\hfill }

\markboth{\hfill{\footnotesize\rm Yang Feng and Yi Yu} \hfill}
{\hfill {\footnotesize\rm Leave-$n_v$-out CV for high-dimensional variable selection} \hfill}

\renewcommand{\thefootnote}{}
$\ $\par


\fontsize{10.95}{14pt plus.8pt minus .6pt}\selectfont
\vspace{0.8pc}
\centerline{\large\bf The restricted consistency property of leave-$n_v$-out}
\vspace{2pt}
\centerline{\large\bf  cross-validation for high-dimensional variable selection}
\vspace{.4cm}
\centerline{Yang Feng and Yi Yu}
\vspace{.4cm}
\centerline{\it Columbia University and University of Bristol}
\vspace{.55cm}
\fontsize{9}{11.5pt plus.8pt minus .6pt}\selectfont


\begin{quotation}
\noindent {\it Abstract:}
%
%
Cross-validation (CV) methods are popular for selecting the tuning parameter in the high-dimensional variable selection problem.  We show the mis-alignment of the CV is one possible reason of its over-selection behavior.  To fix this issue, we propose a version of leave-$n_v$-out cross-validation (CV($n_v$)), for selecting the optimal model among the restricted candidate model set for high-dimensional generalized linear models.  By using the same candidate model sequence and a proper order of construction sample size $n_c$ in each CV split, CV($n_v$) avoids the potential hurdles in developing theoretical properties.  CV($n_v$)  is shown to enjoy the restricted model selection consistency property under mild conditions.  Extensive simulations and real data analysis support the theoretical results and demonstrate the performances of CV($n_v$) in terms of both model selection and prediction.
\par

\vspace{9pt}
\par
\end{quotation}\par

\def\thefigure{\arabic{figure}}
\def\thetable{\arabic{table}}

\fontsize{10.95}{14pt plus.8pt minus .6pt}\selectfont

\newpage
\lhead[\footnotesize\thepage\fancyplain{}\leftmark]{}\rhead[]{\fancyplain{}\rightmark\footnotesize\thepage}

\setcounter{chapter}{1}
\setcounter{equation}{0} 

\noindent {\bf 1. Introduction}

In recent years, massive high-throughput data sets are generated as a result of technological advancements in many fields.  Such data are featured by the large number of variables $p$  compared with the sample size $n$.  For an overview of the many challenges associated with high-dimensional statistical modeling, we refer the readers to \cite{FanLv2010}  and \cite{BuehlmannGeer2011}. 


A crucial goal in high-dimensional data analysis is to achieve a good balance between the goodness-of-fit and the complexity of the model, as both predictability and model interpretability are important to practitioners in many scientific fields. One popular avenue to achieve this balance is the imposition of penalties on the model complexity, which leads to simultaneous variable selection and parameter estimation in one single step.  Numerous efforts have been made, from both theoretical and numerical perspectives; to name but a few, \cite{Tibshirani1996} proposed Lasso, which is the $\ell_1$ penalty, or equivalently, \cite{ChenDonoho1994} proposed basis pursuit.  Also, folded-concave penalties including SCAD \citep{FanLi2001} and MCP \citep{Zhang2010} have been proposed and widely used over the years.  


One of the important aspects of penalization techniques is the tuning parameter, which determines how much penalty is imposed.  Over-penalization runs the risk of overlooking scientifically meaningful information; on the other hand, under-penalization may erroneously identify seemingly meaningful patterns that are actually the result of experimental noise.  It is, therefore, critical to choose the tuning parameter with care.
 
There has been an abundance of research on using certain kind of information criteria to select the tuning parameter; these include the generalized cross-validation \citep{Tibshirani1996, wang2007tuning}, the $C_p$ \citep{efron2004least}, the extended Bayesian information criterion (EBIC) \citep{ChenChen2008, luo2014sequential}, the modified BIC \citep{wang2009shrinkage}, the generalized information criterion \citep{ZhangEtal2010, Fan.Tang.2012}, etc.  Other related work includes selection of the tuning parameter through joint estimation of the regression coefficient and the standard deviation \citep{Stadler.Bulmann.ea.2010, Sun.Zhang.2011}. 

Another popular method for selecting the tuning parameter is cross-validation (CV), which is a data-driven method.  A vast amount of theoretical work has been done for CV in the fixed-dimensional linear regression models.  For example, leave-one-out CV (CV(1)) is shown to be asymptotically equivalent to Akaike information criterion (AIC), $C_p$, jackknife, and bootstrap \citep{Stone1977, Efron1983, Efron1986}.  \cite{Shao1993} proved the model selection inconsistency of CV(1) for the fixed-dimensional linear regression model.  In addition, for leave-$n_v$-out CV (CV($n_v$)),  the author gave the proper ratio of the size of construction set to that of validation set which turns out to be necessary for model selection consistency.  Here, by construction and validation data sets we mean the subsets used to construct and validate the estimators in CV splits.  However, $K$-fold CV, the most commonly used method, is well known for its conservativeness, i.e. the corresponding estimator selects far too many noise variables \citep{yu2014modified}.  As mentioned in \cite{ZhangHuang2008}, the theoretical justification of CV based tuning parameter is unclear for model selection purposes. \cite{yu2014modified} proposed  the modified cross-validation for high-dimensional linear regression models and showed that it outperforms the regular $K$-fold CV in numerical experiments. Compared with \cite{yu2014modified}, in this paper, we study the leave-$n_v$-out cross-validation for a sequence of candidate models from the whole data set and developed the restricted consistency results under the generalized linear model framework for high-dimensional variable selection. 

Another related work is relaxed Lasso \citep{Meinshausen},  which is a two-stage method, with the penalty at the second stage only operating on those variables being selected at the first stage.  The author conjectured that the $K$-fold CV for this two-step method will achieve model selection consistency. Different from \cite{Meinshausen}, in this paper we study the theoretical behaviors of the penalties, and mainly focus on the model selection instead of proposing a variant of Lasso procedure, with rigorous discussions on the asymptotic behavior of CV provided. 


The main contribution of the paper is two-fold: (1) investigations are conducted for the advantages and drawbacks of the commonly used CV methods for tuning parameter selection in the penalized estimation methods; (2) studying the leave-$n_v$-out cross-validation, which is shown to be consistent in a restricted sense for a wide range of penalty functions in the high-dimensional generalized linear model framework.  

We would like to introduce some of the notation used throughout this paper. For a $p$-dimensional vector $\boldsymbol{\beta}$ and an $n\times p$-dimensional matrix $A$, suppose $s$ is a subset of $\{1,\cdots,n\}$ and $\alpha$ is a subset of $\{1,\cdots, p\}$, then $\boldsymbol{\beta}_{\alpha}$ represents the subvector of $\boldsymbol{\beta}$ corresponding to $\alpha$, $A_s$ represents the submatrix of $A$ corresponding to rows with indices in $s$ and $A^{\alpha}$ represents the submatrix of $A$ corresponding to columns with indices in $\alpha$.  Let $|s|$ represent the cardinality of set $s$. In addition, define the $\ell_0$, $\ell_1$ and $\ell_2$ norms of $\boldsymbol{\beta}$ as $\|\boldsymbol{\beta}\|_0=\sum_{j=1}^p 1\{\beta_j\neq 0\}$, $\|\boldsymbol{\beta}\|_1=\sum_{j=1}^p |\beta_j|$ and $\|\boldsymbol{\beta}\|=[\sum_{j=1}^p\beta_j^2]^{1/2}$, respectively.  Let $g_1$ and $g_2$ be two functions of $n$. We use $g_1(n) = \Theta(g_2(n))$ to represent that they are asymptotically of the same order, i.e., there exist positive constants $c_1$ and $c_2$, such that 
	\[
		c_1 \le \liminf_n g_1(n)/g_2(n) \le \limsup_n g_1(n)/g_2(n) \le c_2.
	\]


\par

The rest of the paper is organized as follows. We introduce the generalized linear model setup and discuss the $K$-fold CV in Section 2. Motivated by the issues with the $K$-fold CV, we introduce the leave-$n_v$-out cross-validation (CV($n_v$)) for high-dimensional variable selection, and show CV($n_v$) can achieve the restricted model selection consistency in Section 3. We conduct extensive simulation studies in Section 4 and a real data analysis in Section 5 to compare CV($n_v$) with other types of CV methods as well as many state-of-the-art information criteria. We conclude the paper with a short discussion in Section 6 while all the technical details are collected in the Appendix.

\noindent {\bf 2. Model setup and $K$-fold cross-validation}

\noindent {\bf 2.1. Model setup}

Suppose we have $n$ i.i.d. observation pairs $(\boldsymbol{x}_i, y_i)$, $i = 1, \cdots, n$, where $\boldsymbol{x}_i$ is a $p$-dimensional predictor and $y_i$ is the response.  For generalized linear models, we assume the conditional distribution of $y$ given $ \boldsymbol{x}$ belongs to an exponential family with the canonical link and canonical parameter $\theta = \boldsymbol{x}^{\top}\boldsymbol{\beta}$; that is, it has the following density function,
\begin{align*}
f(y; \boldsymbol{x}, \boldsymbol{\beta}) = c(y,\phi)\exp ((y\theta - b(\theta))/a(\phi)),
\end{align*}
where $\phi \in (0, \infty)$ is the dispersion parameter, the functions $a(\cdot)$, $b(\cdot)$ and $c(\cdot, \cdot)$ are known and different for different models.  Let $\boldsymbol{\beta}^o$ be the true regression parameter, with $\|\boldsymbol{\beta}^o\|_0 = d_o$.  In the high-dimensional setting, $p$ may well exceed $n$ but $d_o$ is usually assumed to be strictly upper bounded by $n$, i.e., $d_o < n$.  Up to an affine transformation with $\theta_i = \boldsymbol{x}_i^{\top}\boldsymbol{\beta}$, the log-likelihood divided by the sample size is given by
\begin{align}\label{eq::log-lik}
\ell(\boldsymbol{\beta}) = n^{-1}\sum_{i=1}^n \{y_i\theta_i - b(\theta_i)\}.
\end{align}
Minimizing the penalized negative log-likelihood function leads to the following estimator.
\begin{align}\label{eq::pen-log-lik}
	\hat{\boldsymbol{\beta}}(\lambda) = \argmin_{\boldsymbol{\beta} \in \mathbb{R}^p} \{-\ell(\boldsymbol{\beta}) + p_{\lambda, \gamma}(\boldsymbol{\beta})\},
\end{align}
where $p_{\lambda, \gamma}(\cdot) $ is the penalty function. 

Given subset $s\subset \{1,\cdots, n\}$,   the  log-likelihood function evaluated on the subset $s$ is
\begin{align}\label{eq::log-lik-s}
	\ell^{(s)}(\boldsymbol{\beta}) = (|s|)^{-1}\sum_{i\in s} \{y_i\theta_i - b(\theta_i)\}.
\end{align}
Then the corresponding minimizer of the penalized negative log-likelihood is
\begin{align}\label{eq::pen-log-lik-s}
	\hat{\boldsymbol{\beta}}^{(s)}(\lambda) = \argmin_{\boldsymbol{\beta} \in \mathbb{R}^p} \{-\ell^{(s)}(\boldsymbol{\beta}) + p_{\lambda, \gamma}(\boldsymbol{\beta})\},
\end{align}

In this paper, we only consider separable sparsity-inducing penalties; that is, there exists a non-negative function $\rho(\cdot)$, such that for any vector $\boldsymbol{\beta} = (\beta_1, \cdots, \beta_p)^{\top}$, the penalty function $p_{\lambda, \gamma}(\cdot)$ satisfies
\begin{align}\label{sep}
p_{\lam, \gamma}(\bbeta) = \sum_{j=1}^p \rho(|\beta_j|; \lam, \gamma),
\end{align}
where $\lam$ and $\gamma$ are the parameters of the penalty function and the minimizer of the penalized negative log-likelihood leads to a sparse solution.  Both convex and folded-concave penalties can be written in the form of \eqref{sep}.  For convex penalties, such as Lasso \citep{Tibshirani1996}, $\gamma = \infty$; while for folded-concave penalties, $0 < \gamma < \infty$.  In the penalty function \eqref{sep}, $\gamma$ is  a  parameter controlling the concavity of the penalty, and  in this paper we only focus on studying the collection of solutions as $\lambda$ changes while fixing $\gamma$.

 A popular class of algorithms on solving \eqref{eq::pen-log-lik} are called path algorithms. Many different path algorithms have been proposed, including  forward regression, stepwise regression, \pkg{lars} \citep{efron2004least}, \pkg{glmpath} \citep{ParkHastie2007}, \pkg{glmnet} \citep{FriedmanHT2010}, \pkg{ncvreg} \citep{BrehenyHuang2011}, \pkg{apple}  \citep{yu2014apple}, among others.   In a  path algorithm, a collection of (usually sparse) estimators  $\{\hat\bbeta_r, r=1,\cdots,R\}$ are generated, where $R$ represents the total number of candidate estimators.    Then the natural question is to choose the best estimate $\hat\bbeta_{\hat r}$  among the $R$ candidates according to certain criteria. 

\noindent {\bf 2.2. Cross-validation}

There are many different versions of CV, to avoid ambiguity, we describe $K$-fold CV with \pkg{glmnet} and \pkg{ncvreg} in the penalized negative log-likelihood context in Algorithm \ref{al:Kfcv}.

\begin{algorithm}[t]
	
\caption{$K$-fold CV for a typical path algorithm.}\label{al:Kfcv}
\begin{enumerate}
	\item[S1.] Using the whole data, generate a data-driven penalty parameter sequence $\blambda = \{\lambda_1, \cdots, \lambda_R\}$. Compute the solution path  $\{\hat{\boldsymbol{\beta}}_r, \, r= 1, \cdots, R\}$, where $\hat\bbeta_r = \hat\bbeta(\lambda_r)$. 
	\item[S2.] Randomly divide the data set into $K$ folds, and denote the indices of each fold as $s_k$, $k = 1, \cdots, K$,  and $s_{(-k)} =  \{1,\cdots, n\} \setminus s_k$.
	\item[S3.] For each fold $k=1,\cdots, K$
	\begin{enumerate}
	\item 
	 Using the construction data in $s_{(-k)}$, generate its own penalty parameter sequence $\blambda^{(-k)} = \{\lambda_1^{(-k)}, \cdots, \lambda_R^{(-k)}\}$. 
	 \item Compute the corresponding solution path $\{\hat{\boldsymbol{\beta}}^{(-k)}_r, \, r = 1, \cdots, R\}$, where $\hat\bbeta^{(-k)}_r = \hat\bbeta^{(s_{(-k)})}(\lambda^{(-k)}_r)$ is the penalized estimator defined in \eqref{eq::pen-log-lik-s} with penalty parameter $\lambda_r^{(-k)}$.
	 \item Evaluate the prediction performance of $\{\hat{\boldsymbol{\beta}}^{(-k)}_r, \, r = 1, \cdots, R\}$ on the validation data in  $s_k$ using the negative log-likelihood function. The resulting values are denoted by $\{L^k_r, r=1,\cdots, R\}$, where $L^k_r = -\ell^{(s_k)}(\hat\bbeta^{(-k)}_r)$ as defined in \eqref{eq::log-lik-s}.
		\end{enumerate}
	\item[S4.] Calculate the average criterion values $\{L_r, \, r=1,\cdots, R\}$ where $L_r = K^{-1}\sum_{k=1}^K L^k_r$. Output the optimal location $\hat r = \arg\min_{r =1, \ldots, R} L_r$ along with its corresponding solution $\hat\bbeta_{\hat r}$. 
\end{enumerate}

\end{algorithm}

In Algorithm~\ref{al:Kfcv}, to compare the performance of $\{\hat{\boldsymbol{\beta}}_r, r = 1,\cdots, R\}$, we are averaging the prediction performance on the corresponding validation set $s_k$ over the $K$ folds using the estimator $\hat\bbeta_r^{(-k)}$ from the construction set $s_{(-k)}$. However, there is no guarantee that we are averaging across the same models or the same tuning parameters across different folds.  In path algorithms, the tuning parameters are determined by the construction data set, and the estimators are determined by the tuning parameters and the construction data set. 

\begin{remark}
In some other path algorithms including \pkg{lars} and \pkg{glmpath}, instead of starting with a sequence of data-driven penalty parameters, they proceed by adaptively adding/deleting one predictor at a time from the model and provide the corresponding $\hat\bbeta_r$ after each operation. Note that the solution $\hat\bbeta_r$ from such path algorithms would also implies a certain value of $\lambda_r$ in \eqref{eq::pen-log-lik}. As a result, the preceding discussions regarding the averaging process also apply to those algorithms. 
\end{remark}

\begin{remark}
In practice, it is also common to use the tuning parameter sequence $\boldsymbol{\lambda}$ generated by the whole data in all splits.  Although this guarantees the alignment of tuning parameters across different splits, this still causes mis-alignment in terms of model sequences, and this actually could cause more problems due to the fact that a desirable tuning parameter should be a function of the sample size.  In any case, it would be very difficult to link the chosen tuning parameter from the splits with its performance in the whole data.
\end{remark}

We conduct a simple simulation for a high-dimensional linear regression example with the 5-fold CV. In Figure~\ref{fig:ae1}, we show the corresponding results of two different construction data sets when performing the CV.  In the left panel we show the first 30 values of $\lambda$ on each path, with the $x$-axis being the location indices; in the right panel, the sequences of the model sizes are presented against their locations on the solution paths.  The CV averages models across different splits, but as we can see from Figure~\ref{fig:ae1}, the corresponding $\lambda$ sequences and the model size sequences are both very different for those two different splits. As a result, it is very difficult to derive any theoretical justification for either model selection or tuning parameter selection property of the CV tuned estimator.  More numerical results on the alignment issue of the CV will be shown in Section~4.1.

\begin{figure}[h]
\centering
	\includegraphics[width = \textwidth]{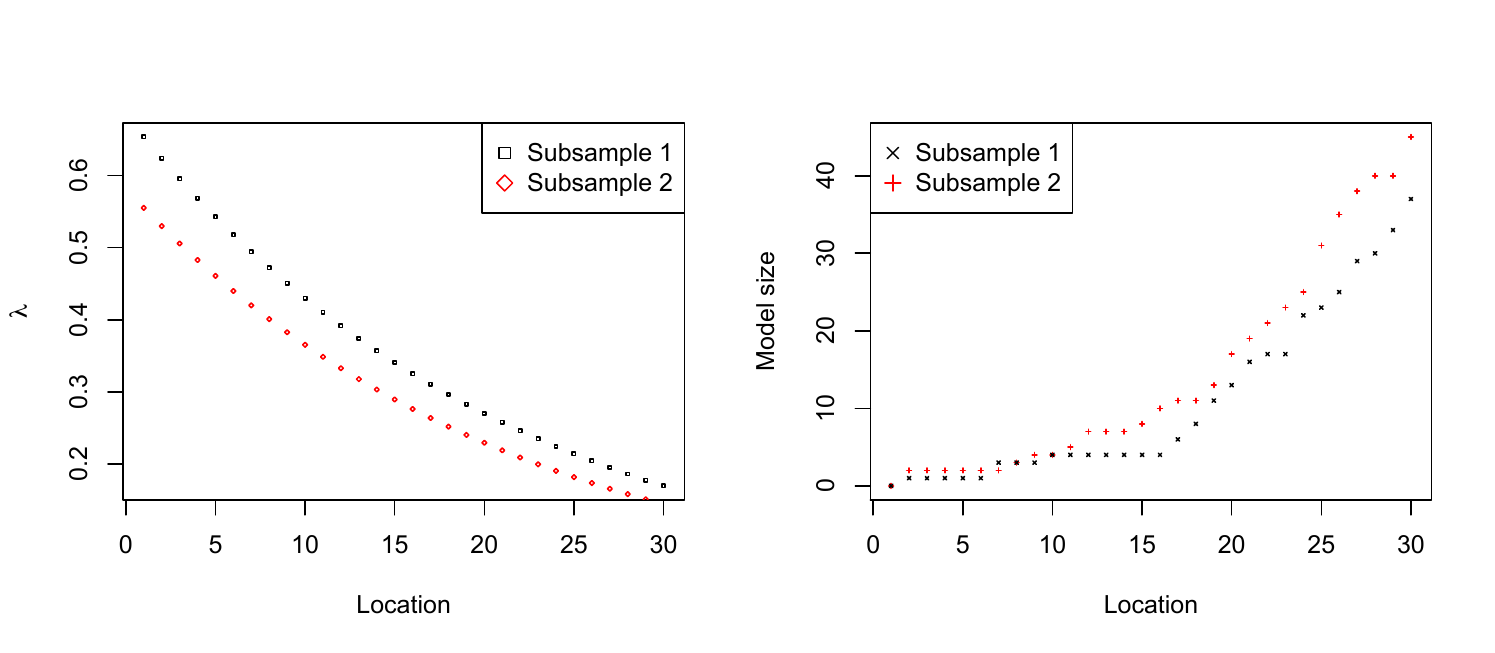}
	\caption{An example on $K$-fold cross-validation}
	\label{fig:ae1}
\end{figure}

\setcounter{chapter}{3}
\setcounter{equation}{0} 
\noindent {\bf 3. Leave-$n_v$-out cross-validation}\\
With a better understanding of the  issues of the CV in Section~2,  we propose a version of CV($n_v$) in this section.  We first introduce some key concepts regarding model selection in Section 3.1 and CV($n_v$), then point out its major differences from the CV in Section 3.2. In Section 3.3, we show that CV($n_v$) is {\it restricted model selection consistent} (to be defined formally in Section 3.1)   under mild technical conditions in the generalized linear model framework for both convex and folded-concave penalties.

\noindent {\bf 3.1. Key concepts} \label{subsec-key}

From the solution path $\{\hat\bbeta_r,\, r=1,\cdots, R\}$ in Algorithm 1, one can get a corresponding path of models $\cA = \{\alpha_r,\, r=1,\cdots, R\}$, where $\alpha_r = \{j\in\{1,\cdots,p\}: \,(\hat\beta_r)_j\neq 0\}$ is the indices with nonzero coefficient estimates. 
Similar to \cite{Shao1993}, we divide $\cA$ into two disjoint subsets: $\cA_c$ and its complement $\cA \setminus \cA_c$, where $\cA_c = \{\alpha \in \cA: \,(X^{\alpha})\boldsymbol{\beta}_{\alpha}^o = X\boldsymbol{\beta}^o\}$.   Here we give three definitions which constitute the fundamental concept of this paper. 

\begin{definition}[True model]\label{def-true-model}
The true model is defined as $\cO = \{j:\, \beta_j^o\neq 0\}$.
\end{definition}

Here,  for any estimated model $\hat\cO$, we define its false negative (FN) to be $|\cO\setminus \hat\cO|$ and its false positive (FP) to be $|\hat\cO\setminus \cO|$. 
Then, for the models in $\cA_c$, $\textsc{FN} = 0$; for the models in $\cA\setminus\cA_c$, $\textsc{FN} > 0$. 

\begin{definition}[Optimal model set]\label{def-optimal}
Let $d_* = \min_{\alpha \in \mathcal{A}_c}|\alpha|$.  Define the optimal model set as $\alpha_* = \{\alpha \in \mathcal{A}_c:\,|\alpha| = d_*\}$.
\end{definition}

When $|\alpha_*| = 1$, there is only one optimal model and with slight abuse of notation, we call $\alpha_*$ as the optimal model. The optimal models can be different from the true model. They are the sparsest models without false negatives. 

\begin{remark}
For any model $\alpha \in \mathcal{A}$, define its fitted risk as follows:
\begin{align*}
	R(\alpha) = \sup_{\boldsymbol{x}\in \mathbb{R}^p: \|\boldsymbol{x}\| = 1}(\boldsymbol{x}_{\alpha}^{\top}\boldsymbol{\beta}_{\alpha}^o-\boldsymbol{x}^{\top}\boldsymbol{\beta}^o)^2 = \bigl\|\boldsymbol{\beta}^o_{-\alpha}\bigr\|^2.
\end{align*}
It is obvious that if $\alpha \in \mathcal{A}_c$, then $R(\alpha) = 0$; otherwise, $R(\alpha) > 0$.  

We now demonstrate the differences between the true model and the optimal model (set) via a toy example. In a linear regression setting, assume the true regression coefficient $\boldsymbol{\beta}^o \in\mathbb{R}^{100}$ being $\beta^o_j = 1$, for $j = 1, \cdots, 5$ and $\beta^o_j = 0$, for $j = 6, \cdots, 100$.  Then, the true model $\mathcal{O} = \{1, 2, 3, 4, 5\}$.  If the candidate models are as follows: $\alpha_1 = \{1,2,3 \}$, $\alpha_2 = \{1,2,3, 4\}$, $\alpha_3 = \{1,2,3, 4, 5, 6\}$ and $\alpha_4 = \{1,2,3, 4, 5, 6, 7\}$; note that the true model is not among the candidate models here.  Both models $\alpha_1$ and $\alpha_2$ miss  at least one important variables, with $R(\alpha_1) = 2$ and $R(\alpha_2) = 1$.  The true model is a subset of both $\alpha_3$ and $\alpha_4$, and $R(\alpha_3) = R(\alpha_4) = 0$.  In this situation, we $\alpha_3, \alpha_4 \in \mathcal{A}_c$.  Recall the definition of the optimal model (set), then we know $\alpha_3$ is the optimal one, since it contains fewer false positive than $\alpha_4$. As a result, it is reasonable to target on the optimal model (set) when the true model is out of reach. 
		

\end{remark}

\begin{definition}[Restricted model selection consistency]\label{def-consistency}
We say that a method has the restricted model selection consistency property, if the selected model $\hat{\alpha}_n$ satisfies 
\begin{align*}
\lim_{n\to \infty} \mathrm{pr}\{\hat{\alpha}_n \in  \alpha_*\} = 1.
\end{align*}
\end{definition}



Here, we do not require any specific path algorithm, but start with a collection of candidate models.  As a result, in Definition~\ref{def-consistency}, by \emph{restricted model selection consistency}, we mean that the selected model is in the optimal model set with probability tending to 1.  This is different from the model selection consistency, which  means $\lim_{n\to \infty} \mathrm{pr}\{\hat{\alpha}_n = \mathcal{O}\} = 1$ in our setup.  However, those two properties coincide when the true model is an available candidate, i.e. $\mathcal{O} \in \mathcal{A}$.

\noindent {\bf 3.2. Methodology}

\begin{algorithm}[t]
	
\caption{CV($n_v$) for a typical path algorithm.}\label{al:CCV}
\begin{enumerate}
	\item[S1.] Compute the solution path  $\{\hat{\boldsymbol{\beta}}_r, \, r= 1, \cdots, R\}$ using a given path algorithm with the whole data. 
	Obtain the sequence of models  $\{\alpha_1,\cdots,\alpha_R\}$, where $\alpha_r$ is the support of $\hat{\boldsymbol{\beta}}_r$.
	\item[S2.] Independently draw validation sets  $\{s_k, \, k=1,\cdots, K\}$, where $s_k\subset \{1,\cdots,n\}$ with $|s_k|=n_v$.  Let $s_{(-k)} =  \{1,\cdots, n\} \setminus s_k$ represent the corresponding construction set $s_{(-k)}$ with $|s_{(-k)}|=n_c$.
	\item[S3.] For each  $k=1,\cdots, K$
	\begin{enumerate}
	\item 
	 Using the construction data in $s_{(-k)}$, compute the collection of solutions $\{\tilde{\boldsymbol{\beta}}^{(-k)}_r, \, r = 1, \cdots, R\}$, where
	  \begin{align}\label{eq::rmle}
	 	\tilde\bbeta^{(-k)}_r=\argmin_{\stackrel{\bbeta\in \mathbb{R}^p,}{ \bbeta_{(-\alpha_r)}= \boldsymbol{0}}}\bigl\{ -\ell^{(s_{-k})}(\bbeta)\bigr\},
	 \end{align}
	 where $\ell^{(s_{(-k)})}(\cdot)$ is defined in \eqref{eq::log-lik-s}.
	 \item Evaluate the prediction performance of $\{\tilde{\boldsymbol{\beta}}^{(-k)}_r, \, r = 1, \cdots, R\}$ on the validation set $s_k$ using the negative log-likelihood function. The resulting values are denoted by $\{L^k_r, \, r=1,\cdots, R\}$, where $L^k_r = -\ell^{(k)}(\tilde\bbeta^{(-k)}_r)$. 
	\end{enumerate}
	\item[S4.] Calculate the average criterion value $\{L_r, \, r=1,\cdots, R\}$ where $L_r = K^{-1}\sum_{k=1}^K L^k_r$. Output $\hat r = \arg\min_{r \in \{1, \cdots, R\}} L_r$ along with its corresponding unpenalized solution $\tilde\bbeta_{\hat r}$ as in \eqref{eq::rmle}.
\end{enumerate}

\end{algorithm}

The detailed algorithm of CV($n_v$) for the high-dimensional penalized regression is elaborated in Algorithm~\ref{al:CCV}.  The main idea is to use the whole data set to derive the collection of solutions and the corresponding model sequence.  The problem of selecting the optimal solution is then reduced to the choice of the optimal model.  In this sense, we recast the tuning parameter selection problem for high-dimensional generalized linear models to the problem of model selection of low-dimensional ones, and the models across different splits are exactly the same, therefore the averaging has intuitive meanings.

Another key ingredient of CV($n_v$) is the choice of $n_c$ and $n_v$, i.e., the sample sizes of the construction and validation subsets. Following \cite{Shao1993, Shao1996}, we choose $n_c$ and $n_v$ such that $n_c/n\to 0$ and $n_c \to \infty$, as $n\to \infty$.  This is different from the $K$-fold CV methods, where a larger proportion of data are used for construction and a smaller proportion for validation.  We would like to briefly explain the intuition of the specific sample splitting.  Note that the purpose of the CV  is to select the best model among the candidates, as a result, besides having an accurate estimation for each model (when $n_c \to \infty$),  perhaps more importantly, we would need a sufficiently large ($n_c/n\to 0$) validation set   to detect the subtle differences among the models.  This is particularly challenging in the high-dimensional settings as there are many possible candidate models. The popular $10$-fold CV, for example, only uses 1/10 of the data for the validation set, which is proved to be too small for the purpose of model selection.

We now present the behavior of CV($n_v$) when $n_v$ varies through a simulation study. In Figure~\ref{fig::nc-rate}, we present the average false positive (FP) and false negative (FN) of CV($n_v$) with a wide range of $n_c$ in linear and logistic regression problems with $n = 500$ and $p = 1000$, with details of setting left in Example \ref{ex::CCV}.  From Figure~\ref{fig::nc-rate}, it is clear that in all cases, the larger the order of $n_c$ is, the more FP are involved, but the less FN. For linear regression, $n_c  = \lceil n^{1/2} \rceil$ has the best performance, while $n_c = \lceil n^{3/4} \rceil$ is the best for logistic regression. The different behaviors for linear regression and logistic regression are due to the fact that when the covariates and the coefficients are the same, the logistic regression needs a larger sample size to fit the model well  compare with linear regression. We provide some intuition as follows. 
Under the canonical link, the Fisher information for generalized linear models can be written as $1/a(\phi)X^{\top}WX$, where $\phi$ is the dispersion parameter. For logistic regression, $W = \diag\{\pi_1(1-\pi_1),\cdots, \pi_n(1-\pi_n)\}$, where $\pi_i = \exp(\boldsymbol{x}^{\top}_i\boldsymbol{\beta})/(1+ \exp(\boldsymbol{x}^{\top}_i\boldsymbol{\beta})) < 1$ in non-degenerate cases; while for linear regression, $W = I_n$. This indicates that logistic regression always has less information than linear regression, which leads to the fact that compared with linear regression, we need a larger sample size for the logistic regression to have the same level of estimation accuracy. 

We conclude that in order to achieve the restricted model selection consistency property, a small $n_c$ rate should be chosen as long as the size of the construction sample is large enough to provide accurate estimates. Despite the comparison above, the optimal $n_c$ rates may change for different settings. In any case, CV($n_v$) with a wide range of $n_c$ values all lead to a better performance than the 10-fold CV as well as AIC and BIC.

\begin{figure}[htbp]
\caption{The average false positive (FP) and false negative (FN) of CV($n_v$) for different $n_c$ values in Example  \ref{ex::CCV}.\label{fig::nc-rate}}
\begin{subfigure}[b]{0.5\textwidth}
\caption{Linear regression, $\rho=0$}
	\includegraphics[scale=0.32]{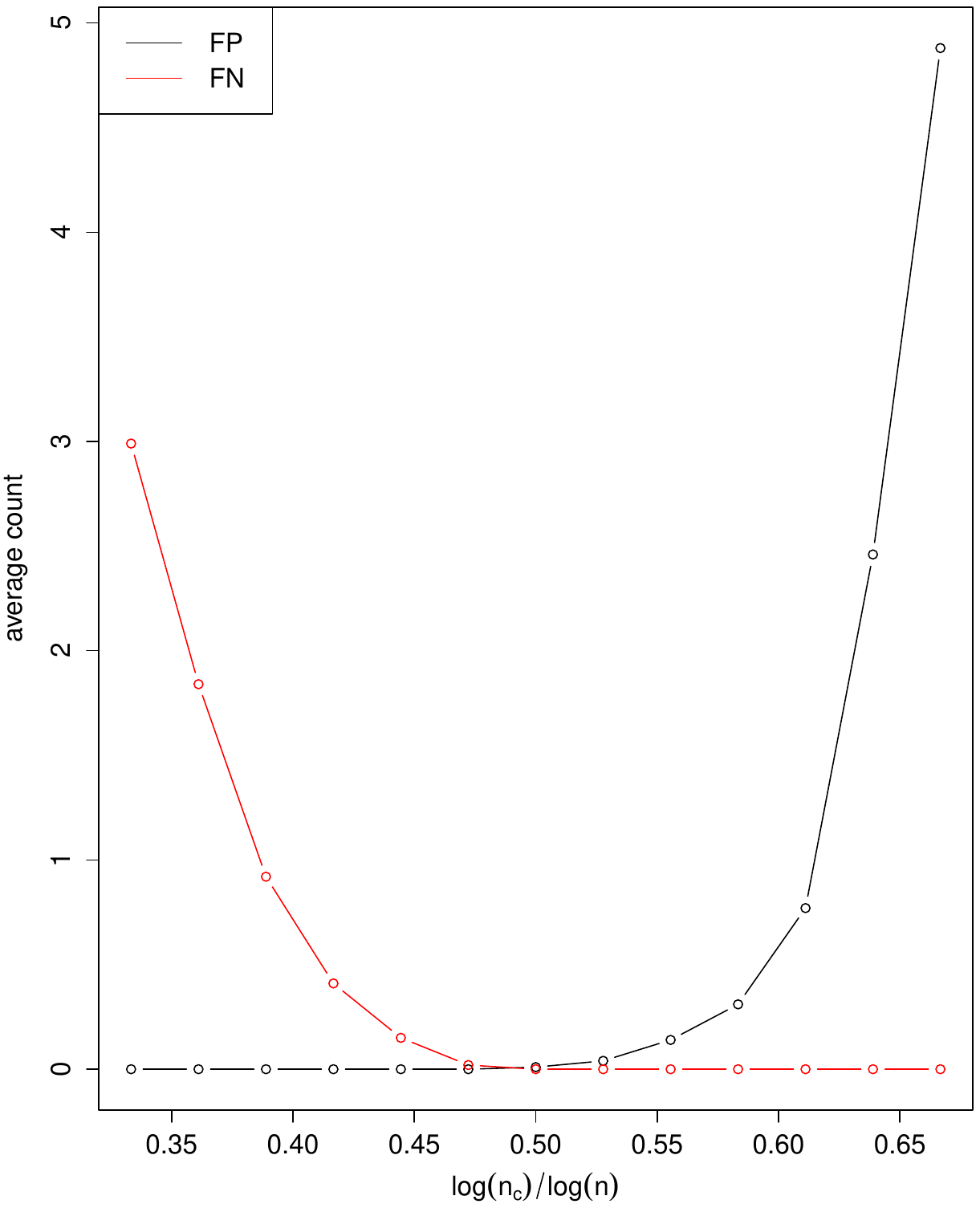}
\end{subfigure}\begin{subfigure}[b]{0.5\textwidth}
\caption{Linear regression, $\rho=0.5$}
	\includegraphics[scale=0.32]{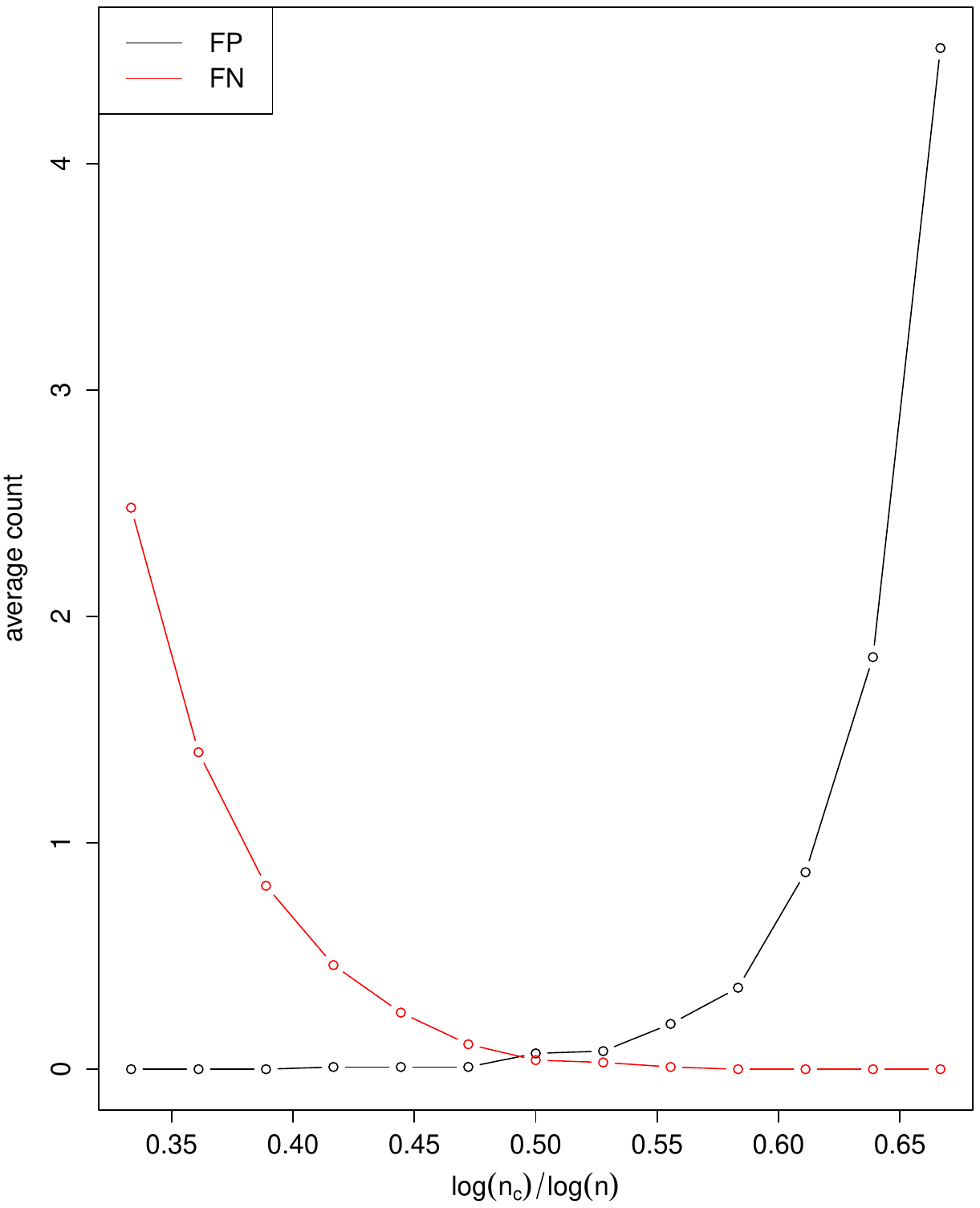}
\end{subfigure}\\
\begin{subfigure}[b]{0.5\textwidth}
\caption{Logistic regression, $\rho=0$}
	\includegraphics[scale=0.32]{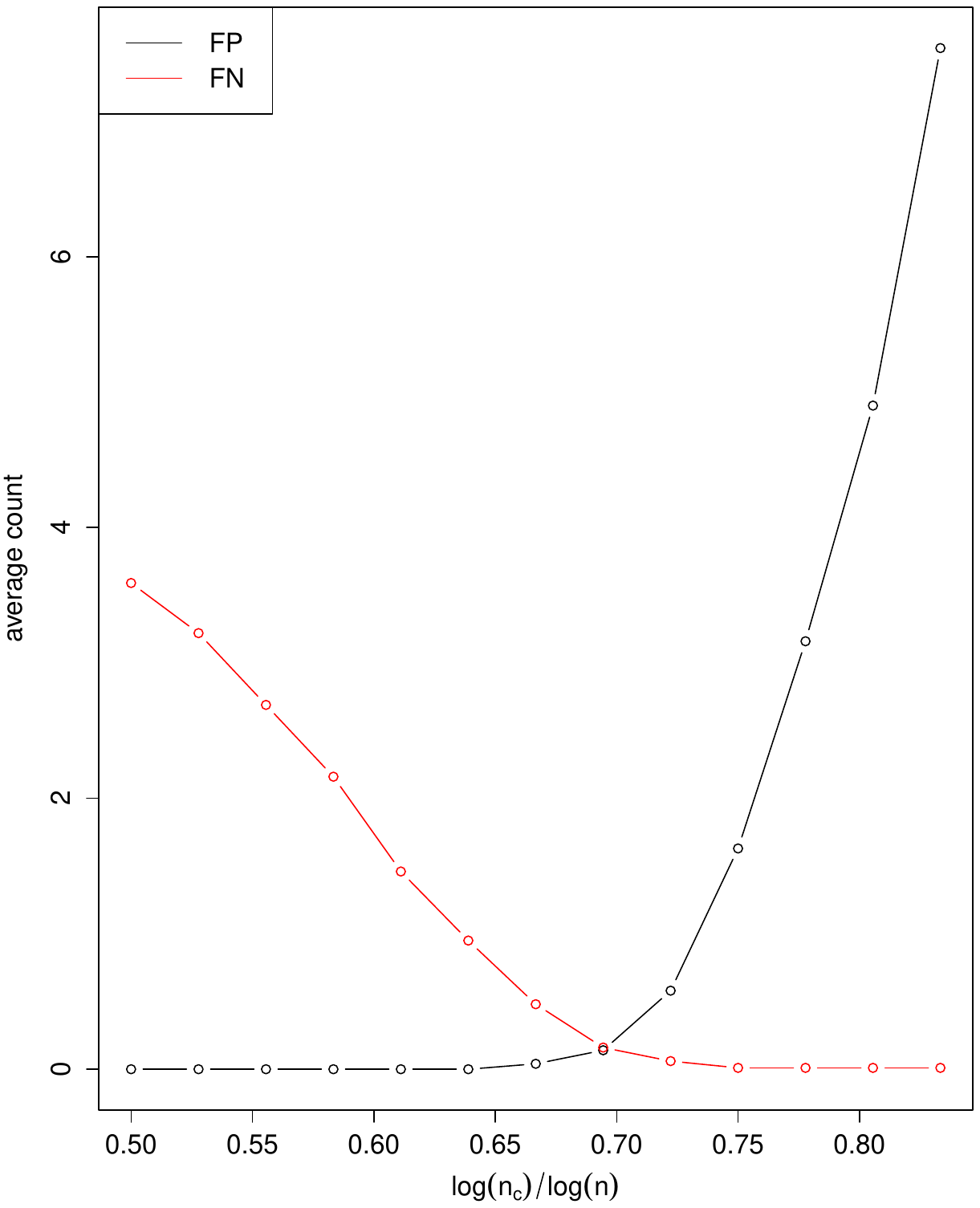}
\end{subfigure}\begin{subfigure}[b]{0.5\textwidth}
\caption{Logistic regression, $\rho=0.5$}
	\includegraphics[scale=0.32]{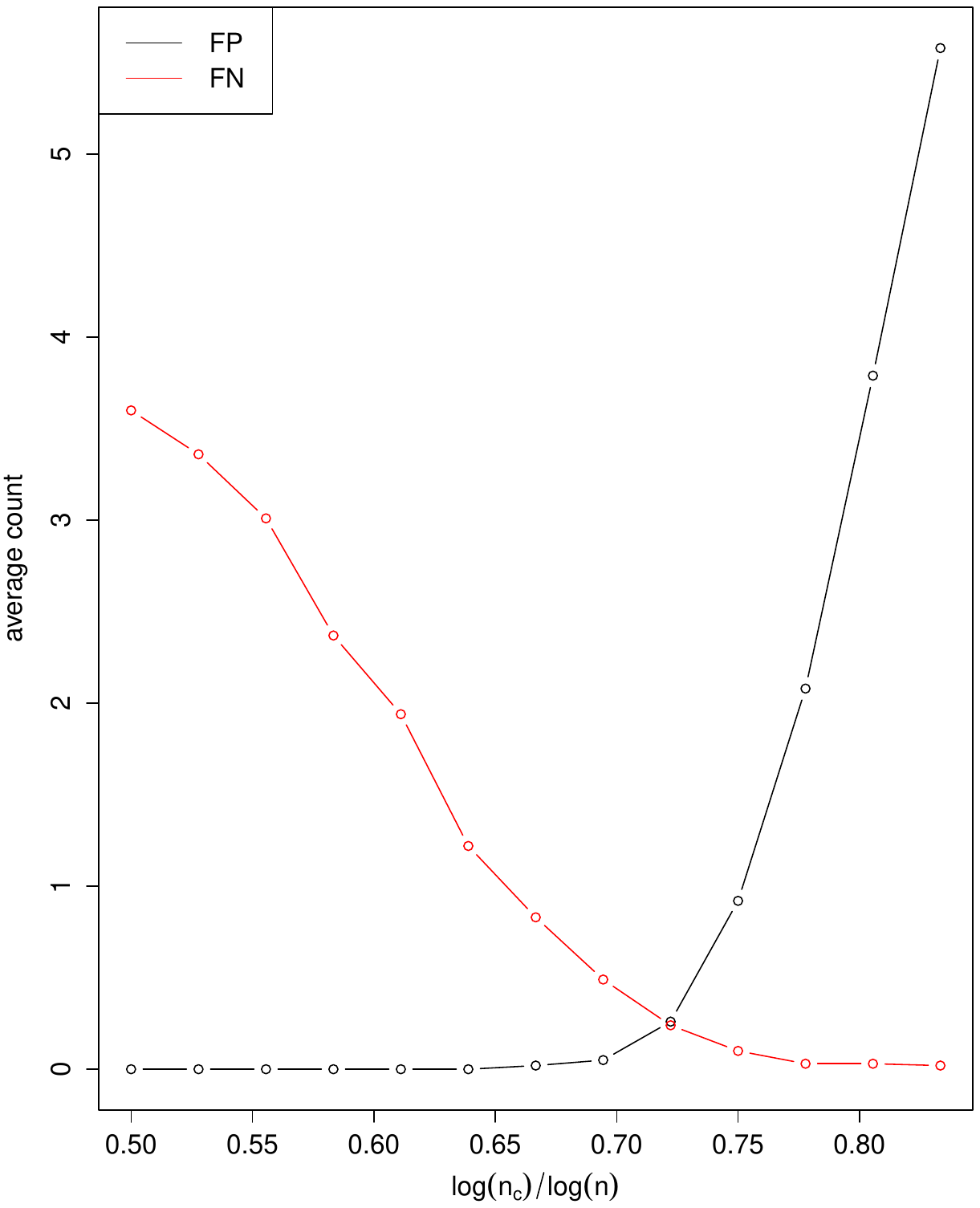}
\end{subfigure}
\end{figure}

Different from \cite{Shao1993, Shao1996}, we are studying a high-dimensional variable selection problem, which leads to fundamental technical differences.  We allow the number of candidate models to diverge, as stated in Condition~\ref{Ac} below, while in \cite{Shao1993, Shao1996} this quantity is a fixed constant.  

\noindent {\bf 3.3. Theory}

Before presenting the theory, we introduce several conditions.

\begin{condition}\label{nonempty}
The set $\mathcal{A}_c$ is not empty.
\end{condition}
Condition \ref{nonempty} is usually satisfied when the penalty parameter $\lambda$ is small enough, and it  ensures that the problem we are trying to solve is not degenerate. 

\begin{condition}[Beta-min]\label{beta} 
For the true model $\cO$, let $\sigma^2 = \mathrm{var}(y)$, we assume\begin{align*}
 \beta_* = \min_{j\in\cO} \left|\beta^o_j\right|\gg \sigma \sqrt{\frac{\log p}{n}}.
\end{align*}
\end{condition}
Condition \ref{beta} is common in the high-dimensional sparse recovery literature, which guarantees that the signal variables are detectable from the noise variables.  If $p = O\left(\exp\left(n^a\right)\right)$, $0<a<1$, then $\beta_* = \Theta(1)$ is sufficient to satisfy this condition.  In fact, $\beta_*$ can go to zero slowly as $n$ and $p$ diverge.

\begin{condition}[Candidate set]\label{Ac} Denote $d_{\max} = \max_{\alpha\in\cA_c}|\alpha|$, $d^* = \max\{d_{\max} - d_*, d_*\}$. Assume $n_c d^* \ll n$ and
\begin{align}\label{N-rate}
R = o\left(\exp\left(n/(n_cd^*)\right)\right).
\end{align}

\end{condition}

Condition \ref{Ac} assures the candidate set is well behaved. The possible size of the candidate set $R$ can diverge at the rate in \eqref{N-rate}. We allow the number of candidate models to diverge as long as $n_c d^* \ll n$.  For instance, if $d^*$ is bounded and $n_c = O(n^{1/2})$, then $R = o(\exp(n^{1/2}))$.  

In the fixed $p$ scenario, the candidate set can be all the possible $2^p$ models.  When we allow both $p$ and $n$ to diverge, we are aware that the number of the candidate models increases too, although in practice this is usually a fixed number, say, $R = 100$ in the default setting in the \pkg{glmnet} package in \pkg{R}.  We can actually control an increasing number of candidate models by exploiting concentration inequalities.  Condition~\ref{Ac} gives the limit of this quantity.

\begin{condition}[Generalized linear models properties]\label{asy}
(i) Assume that $b(\cdot)$ has continuous  first, second and third order derivatives $\dot b(\cdot)$, $\ddot b(\cdot)$ and $\dddot b(\cdot)$; in addition, $\ddot b(\cdot)> 0$; (ii) there exists a function $h(\cdot)$ and $\epsilon_0>0$ such that for any $\alpha\in \cA_c$ and $\boldsymbol{\eta}_{\alpha} \in \{\boldsymbol{\zeta}_{\alpha}: \|\boldsymbol{\zeta}_{\alpha} - \boldsymbol{\beta}_{\alpha}\| \le \epsilon_0\}$, we have  $E(h(\boldsymbol{x})) < \infty$, $E(h_{\alpha}(\boldsymbol{x}_{\alpha})) < \infty$, $\|\ddot b(\boldsymbol{x}_{\alpha}^{\top}\boldsymbol{\eta}_{\alpha})\|^2 \le h_{\alpha}(\boldsymbol{x}_{\alpha})$, and $\|\dddot b(\boldsymbol{x}_{\alpha}^{\top}\boldsymbol{\eta}_{\alpha})\|^2 \le h_{\alpha}(\boldsymbol{x}_{\alpha})$  where $h_{\alpha}(\cdot)$ is the function $h(\cdot)$ restricted to the subspace spanned by $\boldsymbol{x}_{\alpha}$.
\end{condition}

This is a mild condition for generalized linear models. For example, it is easy to verify that the linear regression model satisfies Condition~\ref{asy}, since $b(\theta)=\theta^2/2$, in which case the function $h(\cdot)$ can be set as a constant function.

\begin{condition}[Invertibility condition]\label{IC} 
There exist $c_* > 0$ and $q^* = \Theta(\sqrt{n/\log p})$, such that for all $A \subset \{1, \cdots, p\}$ with $|A| = q^* \ge d_* \ge d_0$, and for any $\boldsymbol{\eta}_A \in \{\boldsymbol{\zeta}_A:\, \|\boldsymbol{\zeta}_A - \boldsymbol{\beta}_A\| \le \varepsilon_0\}$, where $\varepsilon_0 > 0$ is fixed, and if $\boldsymbol{v} \neq \boldsymbol{0}$ is a $q^*$-dimensional vector, we have, 
\begin{align*}
\mathrm{pr}\Bigl\{ c_* \le \bigl\|\bigl(\ddot b (X_A\boldsymbol{\eta}_A)\bigr)^{1/2}X_A \boldsymbol{v}\bigr\|^2/(n\|\boldsymbol{v}\|^2)\Bigr\} \to 1, \quad n\to \infty.
\end{align*}
\end{condition}

This condition indicates that in any manifold of dimension less than or equal to $q^*$, its corresponding restricted MLE is well-defined and unique.  This is a weaker version of the Sparse Riesz Condition \citep{ZhangHuang2008}, in which both upper and lower bounds are required. The Sparse Riesz Condition (or a similar condition) was imposed in the existing literature on the tuning parameter selection consistency using information criteria \citep{ZhangEtal2010}.  With the invertibility condition, we can safely terminate the evaluation on the path when the current model size exceeds $q^*$, without the risk of missing the optimal model. 

\begin{condition}[Design matrix]\label{design}  For all $A \subset \{1, \cdots, p\}$ with $|A| = q^*$, where $q^*$ is defined in Condition \ref{IC}, and for any $\boldsymbol{\eta}_A \in \{\boldsymbol{\zeta}_A: \|\boldsymbol{\zeta}_A - \boldsymbol{\beta}_A\| \le \epsilon_0\}$, where $\epsilon_0 > 0$ is a given constant, the following is satisfied, 
\begin{align*}
\max_{s\in \mathcal{S}} \left\|\frac{1}{n_v} (X_s^A)^{\top}\ddot b (X_s^A\boldsymbol{\eta}_A)(X_s^A) - \frac{1}{n_c} (X_{s^c}^A)^{\top}\ddot b (X_{s^c}^A \boldsymbol{\eta}_A)X_{s^c}^A\right\| = o_p(1),
\end{align*}
where the norm here is the operator norm of matrices, $s^c = \{1, \cdots, n\}\setminus s$ and $\mathcal{S}$ is the collection of splits.
\end{condition}

This condition bounds the difference between the Fisher information of the validation set and the construction set.  This is a reasonably mild condition with the technical details of its corresponding version for linear models  studied in Section 4.4 of \cite{Shao1993}.

\begin{theorem}\label{con}
For the penalized generalized linear models with separable sparse-inducing penalties, assume Conditions~\ref{nonempty}-\ref{design} hold with $n_c/n\to 0$, $n_c\to \infty$,  and the number of the splits $K$ satisfies
\begin{align*}
K^{-1}n_c^{-2}n^2 \to 0.
\end{align*}
Then, CV($n_v$) achieves restricted model selection consistency.
\end{theorem}

In Theorem~\ref{con}, we do not explicitly specify the order of $p$ as a condition, however, the restriction on the dimensionality is implied by Conditions \ref{beta} and \ref{Ac}.  The ultra-high-dimensional setting where  $p = O(\exp(n^a))$, $0<a<1$, is allowed.  Theorem~\ref{con} can be easily derived from Lemma~3 in the Appendix following the fact that CV($n_v$) as described in Algorithm \ref{al:CCV} nails a potentially high-dimensional problem down to a low-dimensional  one.  With the selection consistency in hand, the use of unpenalized solution in S4 of  Algorithm~\ref{al:CCV} is applied to improve the estimation and  prediction performance.

\noindent {\bf 4. Numerical experiments}\\
In this section, we compare the proposed CV($n_v$) with several state-of-the-art tuning parameter selection methods including the $K$-fold CV ($K$-fold), $K$-fold CV with one standard error rule (1SE), AIC, BIC, and EBIC. We study both the linear regression and logistic regression with different correlation structures among covariates.

Before presenting the results of tuning parameter selection, we would like to first examine the behavior of the collections of solutions generated via different splits in a CV procedure. 

\noindent {\bf 4.1. Coherent Rate}

\begin{example}\label{lam-sync}\label{ex::CCV}
(i) Linear regression. For $i = 1, \cdots, n$, let $y_i = \boldsymbol{x}^{\top}_i\boldsymbol{\beta}^o + \varepsilon_i$, where $\boldsymbol{x}_i \stackrel{i.i.d.}{\sim} \mathcal{N}(\boldsymbol{0}_p, \Sigma)$ with $\boldsymbol{0}_p$ the length-$p$ vector with all 0 entries and $\Sigma_{j,k} = \rho^{|j-k|}$, $\varepsilon_i \stackrel{i.i.d.}{\sim} \mathcal{N}(0, 1)$, $\rho = 0.5$, $(n, p) = (500, 10000)$ and $\boldsymbol{\beta}^o \in \mathbb{R}^p$ with the first 9 coordinates (0.8, 0, 0.7, 0, 0.6, 0, 0.5, 0, 0.4) and 0 elsewhere. (ii) Logistic regression. For $i = 1, \cdots, n$, $y_i$ satisfies $\mathrm{pr}(y_i = 1) = \exp(\boldsymbol{x}_i^{\top}\boldsymbol{\beta}^o)/\{1+\exp(\boldsymbol{x}_i^{\top}\boldsymbol{\beta}^o)\} = 1 - \mathrm{pr}(y_i = 0)$, where $\boldsymbol{\beta}^o \in \mathbb{R}^p$ with the first 9 coordinates (1.6, 0, 1.4, 0, 1.2, 0, 1.0, 0, 0.8) and 0 elsewhere. The remaining part of the simulation setting is the same as in (i). 
\end{example}

Suppose the sequence of tuning parameters of the whole data set is $\blambda = (\lam_1,\cdots,\lam_R)$. Here, we take a variant of the 10-fold CV by repeatedly splitting the whole data $K = 100$ times into 9/10 fraction as construction set and the remaining 1/10 fraction as the validation set. Denote the collection of the validation sets as $\{s_k,\, k = 1, \cdots, K\}$ and the construction sets as $\{s_{(-k)},\, k=1,\cdots, K\}$. We also denote $s_0=\{1,\cdots,n\}$ to be the whole sample as a reference.  Denote by $\alpha^{(k)}_{r}$ the model of the $r$-th location in the collection of solutions constructed by subset $s_{(-k)}$ using its corresponding tuning parameter sequence $\blambda^{(-k)}$, where $r = 1,\cdots,R$, $k = 0, 1,\cdots, K$. We define the coherent rate as a sequence representing the degree of agreement of the models across different splits for each tuning parameter location,
\begin{align*}
\textsc{CR}(r) = \bigl| \{k=1,\cdots, K: \, \alpha^{(k)}_{r} = \alpha_{r}^{(0)}\}\bigr|/K, \quad r=1,\cdots,R.
\end{align*}
In the ideal case where $\textsc{CR}(r) = 1$, for all $r = 1, \cdots, R$, the CV method for choosing the tuning parameter may serve as a good surrogate for selecting the optimal model. However,  this is rarely true in practice, especially after the noise variables are activated in the estimators. Next, we demonstrate the behavior of coherent rate. 

For the setting in Example~\ref{ex::CCV}, we calculate the collection of solutions using the \pkg{R} package \pkg{glmnet} for Lasso, and the \pkg{R} package \pkg{ncvreg} for SCAD \citep{FanLi2001} and MCP \citep{Zhang2010}. Figure~\ref{CR} shows how the coherent rate changes along the path in different scenarios. In addition, we mark the location of the 10-fold CV chosen estimator and the first location where noise variables are selected. It is obvious that the coherent rate is much smaller than 1 at most locations.  Note that there exists a small segment where the coherent rate equals 1 for all penalties in linear regression; also, this segment is longer for SCAD and MCP than Lasso. Only on the corresponding segment, the CV is averaging over the same model across different splits. Unfortunately, the 10-fold CV always select a model with coherent rate equals 0 (marked by the solid vertical line in Figure \ref{CR}).   In the logistic regression, all penalties lead to a very small coherent rate even before the noise variables are selected.

From the path-generating procedure, estimators and tuning parameters can be regarded as functions of each other, given the data, so the phenomenon we noticed above is due to the data-driven property of tuning parameters selection. When the data are changed from the whole sample to different subsample splits, the tuning parameter sequence is usually different, and naturally leads to  possibly distinct models. If one wants to hold the models the same, very stringent conditions need to be imposed on the design matrix; these are usually not satisfied even for the simple simulation settings we have shown.

\begin{figure}[htbp]
\caption{The coherent rate along the path for Lasso, SCAD and MCP penalized linear and logistic regression estimators in Example \ref{ex::CCV}. $x$-axis: the location in the collection of solutions, $y$-axis: the coherent rate.  The solid line ``------" is where 10-fold CV chooses, and the dotted line ``- - -" is where noises start to be selected.
\label{CR}}
\includegraphics[scale=0.35]{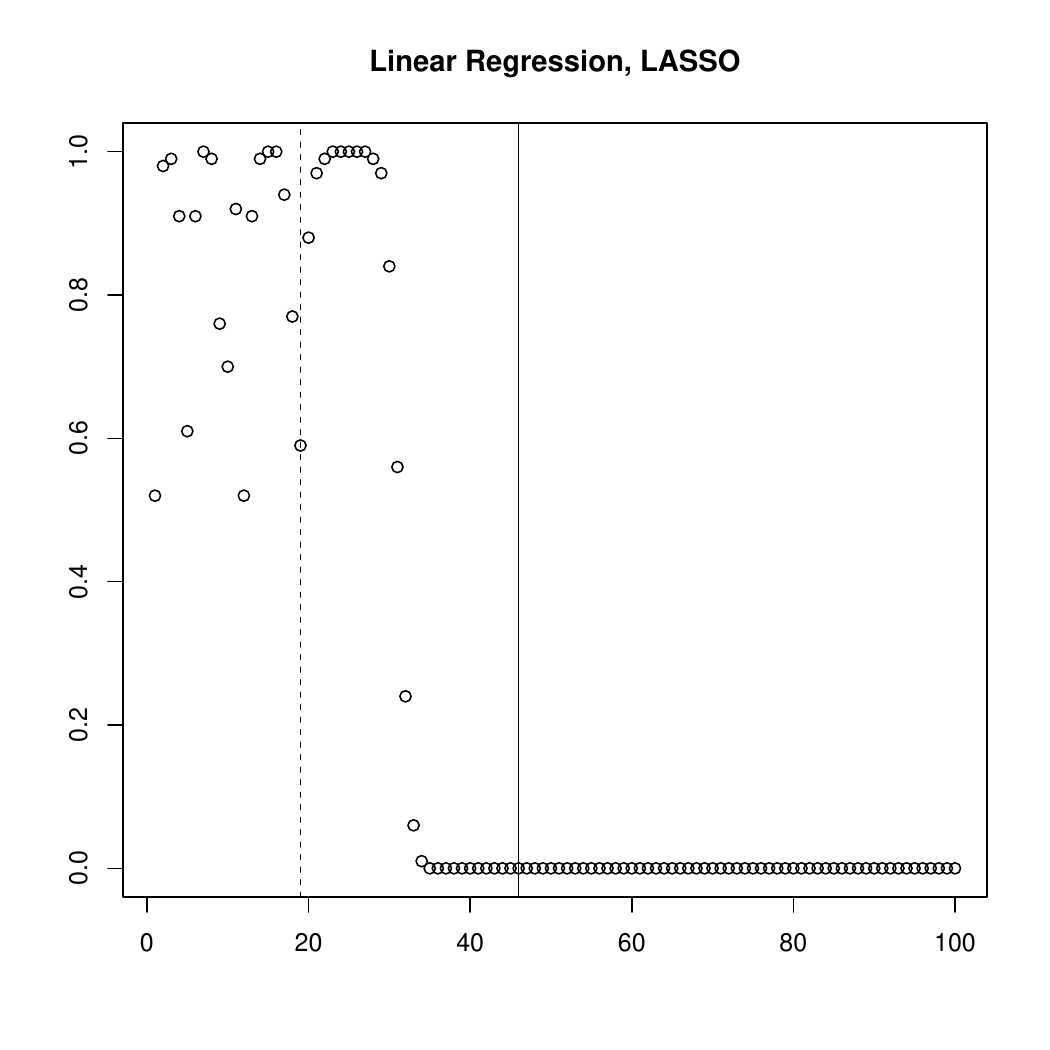}
\includegraphics[scale=0.35]{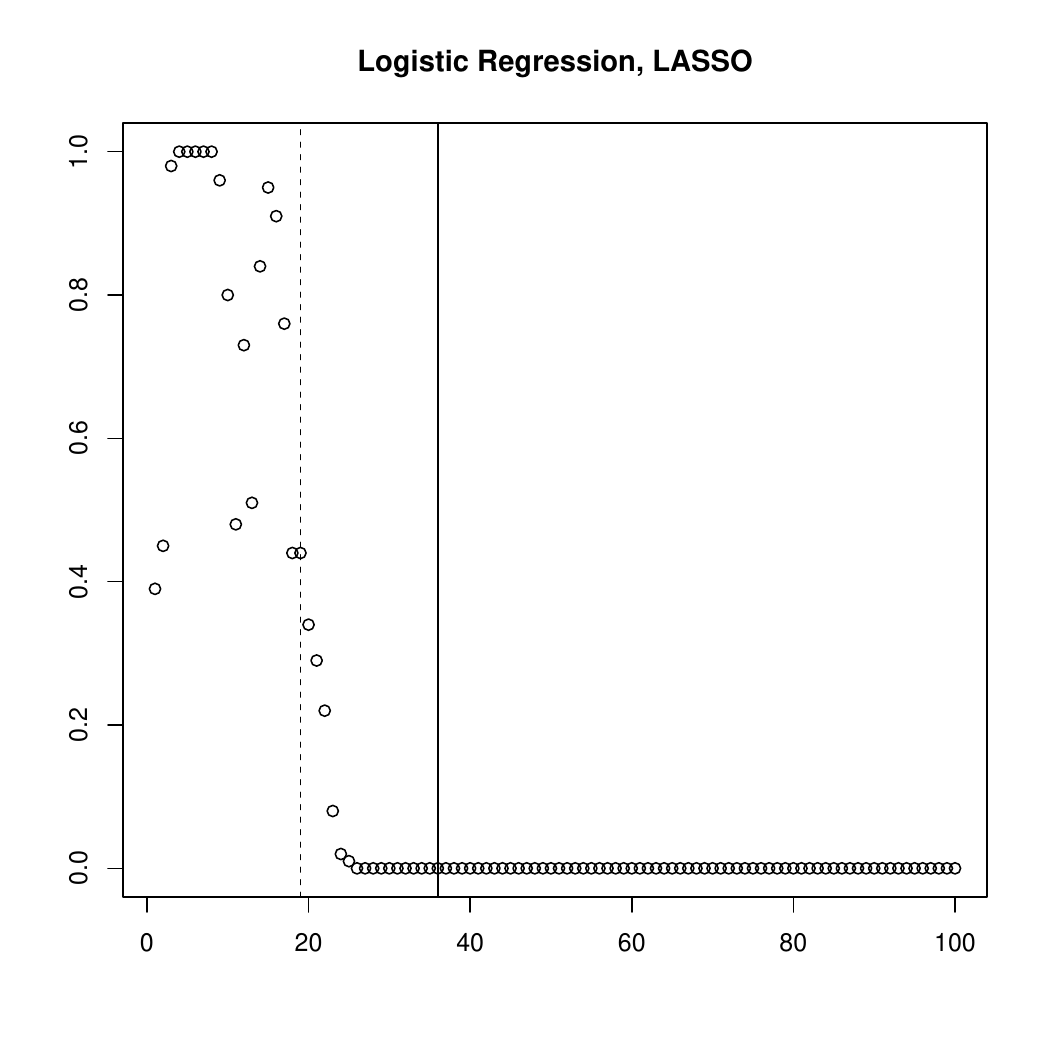}

\includegraphics[scale=0.35]{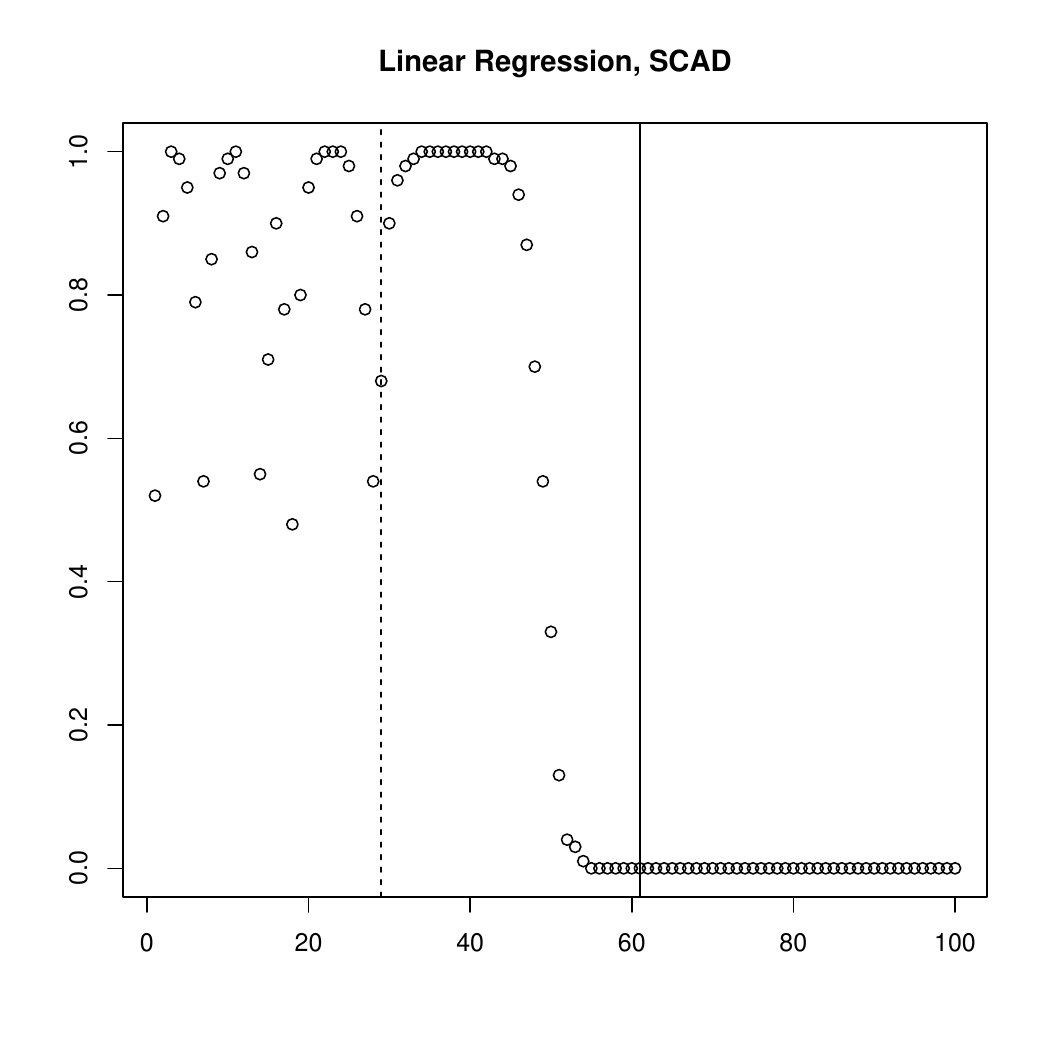}
\includegraphics[scale=0.35]{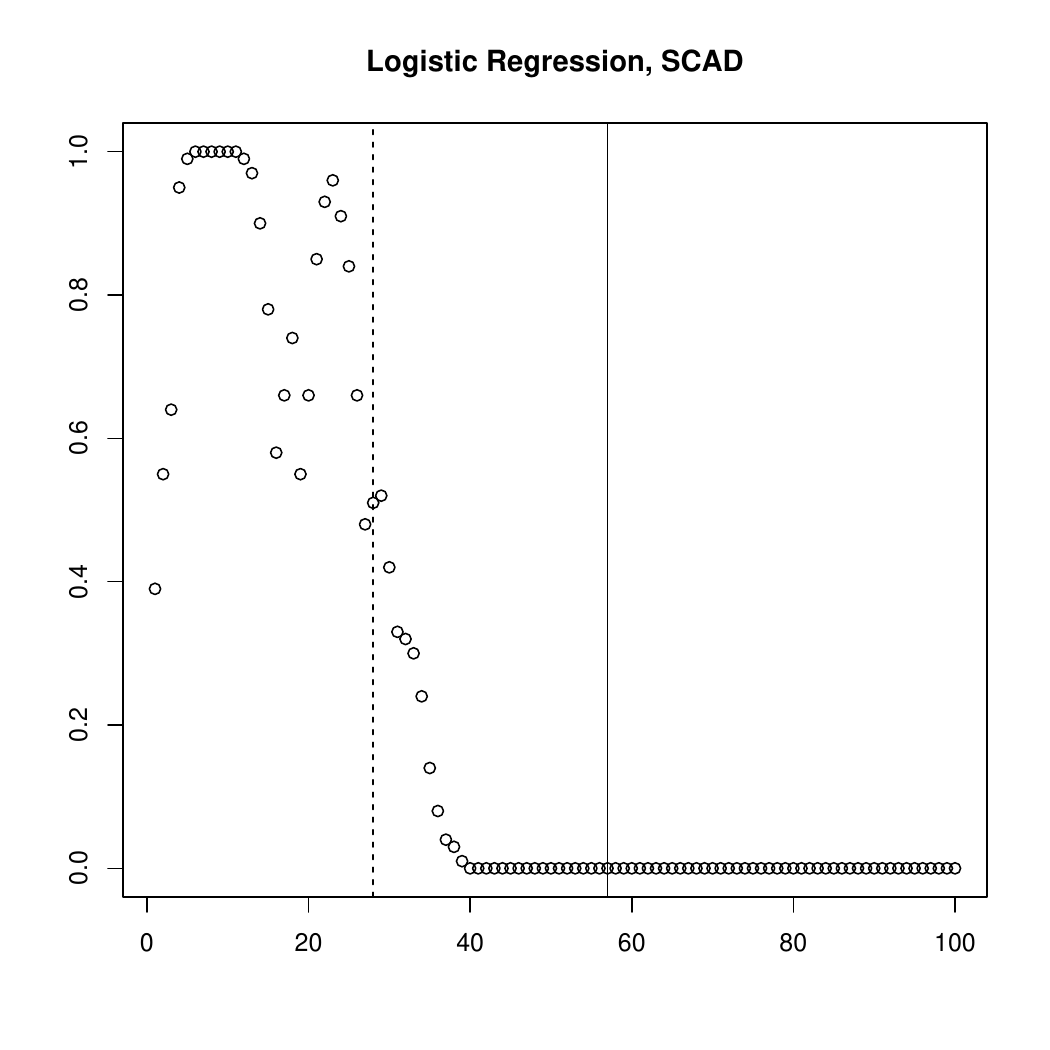}

\includegraphics[scale=0.35]{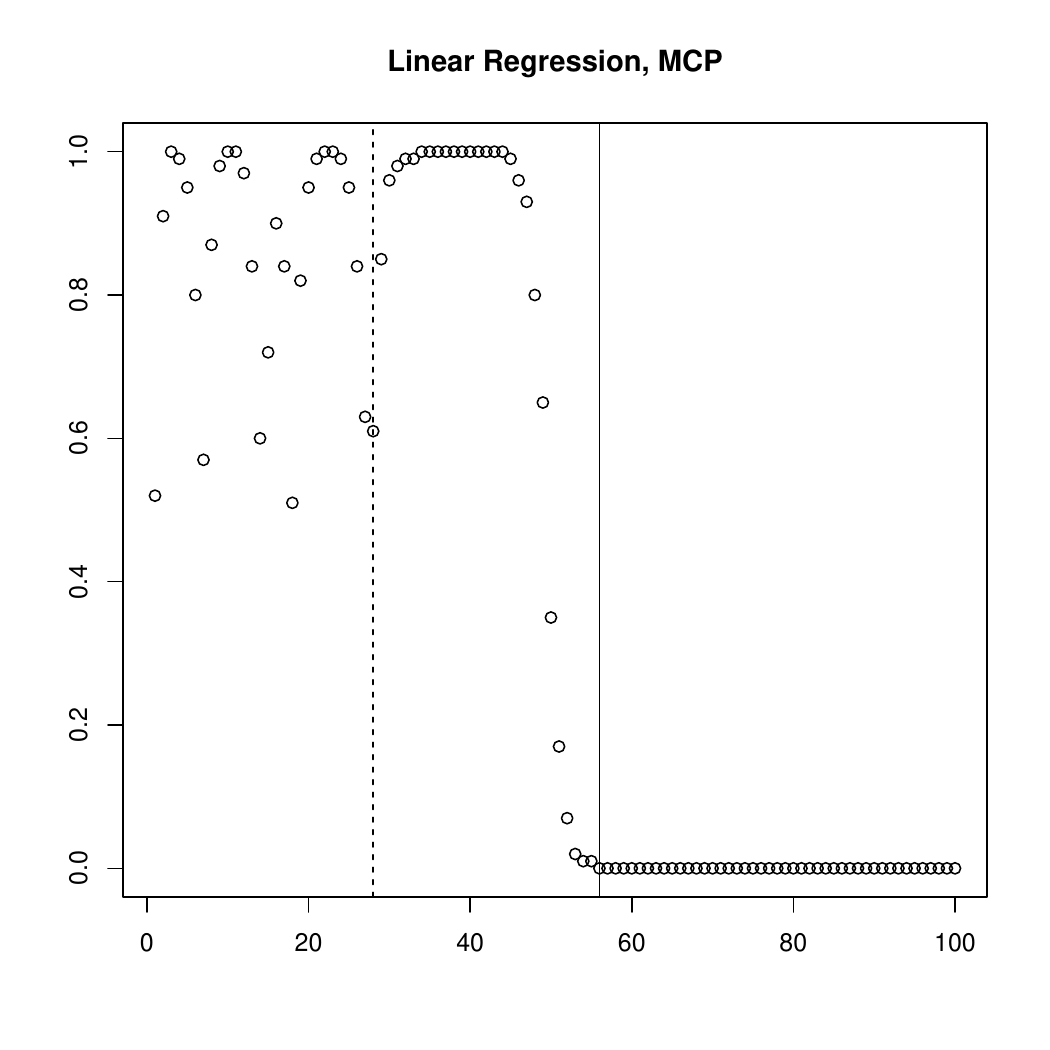}
\includegraphics[scale=0.35]{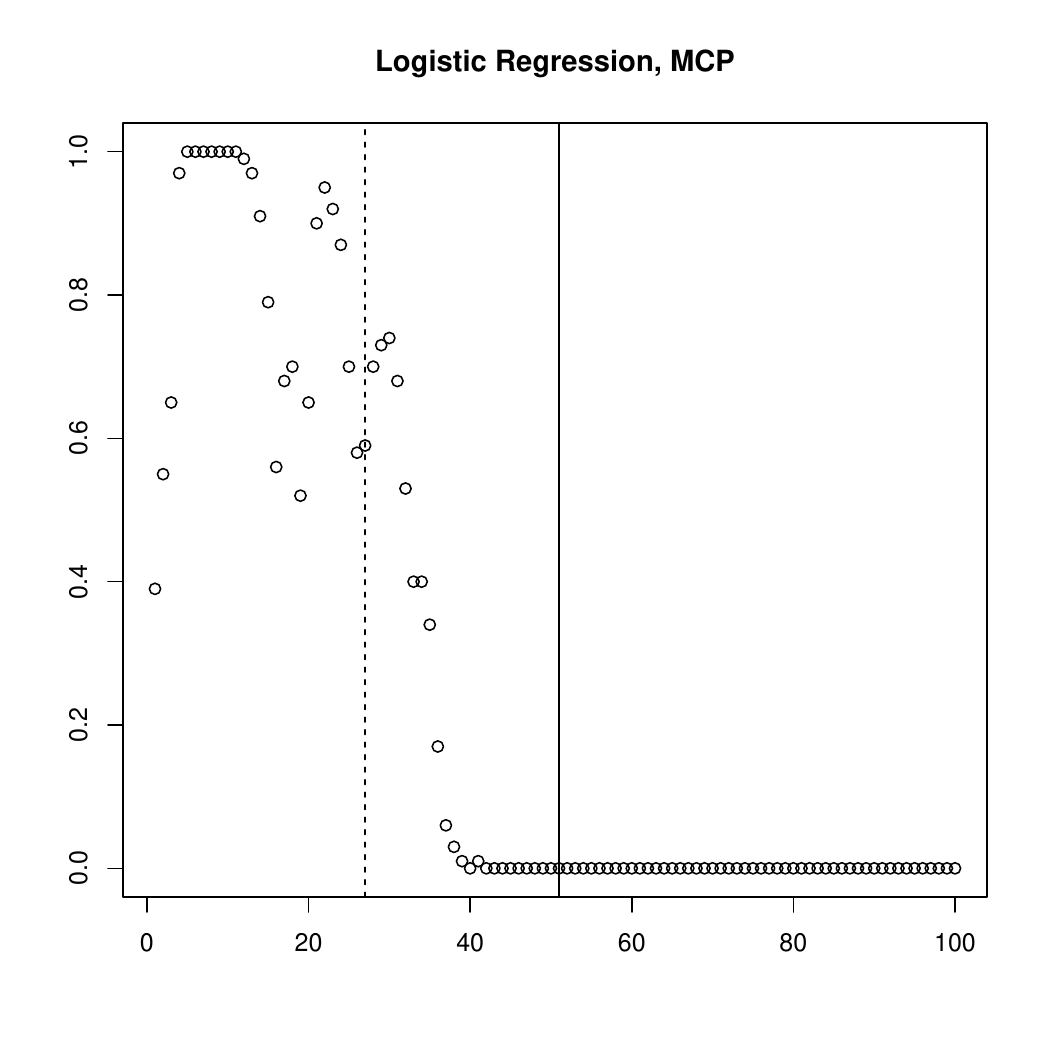}
\end{figure}


\noindent {\bf 4.2. Linear regression}\\
For linear regression, we use the same setting as Example \ref{lam-sync} (i) with $\rho = 0$ and $\rho = 0.5$, and repeat the simulation 100 times. Here the signal to noise ratio (SNR) for the two settings are 1.9 and 4, respectively. 
 For SCAD and MCP paths, we use the default $\gamma = 3$ in the \pkg{ncvreg} package. In Table \ref{T1}, for CV($n_v$), we set $n_c = \lceil n^{1/2} \rceil = 23$ and $n_v = n - n_c = 477$. We compare our results with 10-fold CV in \pkg{glmnet} and \pkg{ncvreg}. We also include the comparison with 10-fold CV with 1SE, where the $\lambda$ is chosen as the maximum one with a loss function value less than the minimum loss function value plus its standard error. In addition, we report the performances of popular information criteria including AIC, BIC and EBIC. To compare these methods, we report false negative (FN), false positive (FP) and prediction error (PE) evaluated on an independent test data set of size $n$. 

\begin{table}[t]
\caption{Comparisons for Example 1(i) with $\rho = 0$ and $\rho = 0.5$ cases. Results are reported in the form of mean (standard error). For CV($n_v$), $n_c  = \lceil n^{1/2} \rceil$ and $K = 50$; FP, false positive; FN, false negative;  PE, prediction error.}
\label{T1}\par
\vskip .2cm
\centerline{\tabcolsep=5truept
\begin{tabular}{|lrrrrrr|}\hline
Method & \multicolumn{3}{c}{$\rho = 0$}& \multicolumn{3}{c|}{$\rho = 0.5$} \\ \hline
Lasso &  FP & FN & PE  & FP & FN & PE \\ \hline
CV($n_v$)&0.01(0.01)&0.00(0.00)&1.01(0.01)&0.07(0.03)&0.04(0.02)&1.02(0.01)\\
10-fold&48.39(3.99)&0.00(0.00)&1.12(0.01)&30.72(3.04)&0.00(0.00)&1.09(0.01)\\
1SE&3.31(0.81)&0.00(0.00)&1.19(0.01)&1.51(0.37)&0.00(0.00)&1.16(0.01)\\
AIC&497.32(1.23)&0.00(0.00)&1.38(0.01)&471.54(1.36)&0.00(0.00)&1.37(0.01)\\
BIC&2.04(0.21)&0.00(0.00)&1.16(0.01)&1.75(0.15)&0.00(0.00)&1.12(0.01)\\
EBIC&0.58(0.08)&0.00(0.00)&1.18(0.01)&0.90(0.09)&0.00(0.00)&1.13(0.01)\\
\hline
SCAD & FP & FN & PE  & FP & FN & PE \\ \hline
CV($n_v$)&0.02(0.01)&0.00(0.00)&1.01(0.01)&0.05(0.02)&0.00(0.00)&1.01(0.01)\\
10-fold&24.50(2.80)&0.00(0.00)&1.03(0.01)&21.74(2.37)&0.00(0.00)&1.03(0.01)\\
1SE&0.48(0.12)&0.00(0.00)&1.08(0.01)&0.21(0.05)&0.00(0.00)&1.08(0.01)\\
AIC&42.19(2.60)&0.00(0.00)&1.03(0.01)&27.02(1.89)&0.04(0.02)&1.07(0.02)\\
BIC&0.94(0.11)&0.00(0.00)&1.04(0.01)&0.77(0.11)&0.04(0.02)&1.08(0.02)\\
EBIC&0.30(0.05)&0.00(0.00)&1.05(0.01)&0.22(0.05)&0.04(0.02)&1.09(0.02)\\
\hline
MCP &  FP & FN & PE  & FP & FN & PE \\ \hline
CV($n_v$)&0.04(0.02)&0.00(0.00)&1.01(0.01)&0.06(0.02)&0.01(0.01)&1.01(0.01)\\
10-fold&4.87(0.66)&0.00(0.00)&1.02(0.01)&5.25(0.66)&0.00(0.00)&1.02(0.01)\\
1SE&0.01(0.01)&0.00(0.00)&1.07(0.01)&0.02(0.02)&0.01(0.01)&1.07(0.01)\\
AIC&87.14(0.45)&0.00(0.00)&1.18(0.01)&80.23(0.75)&0.00(0.00)&1.16(0.01)\\
BIC&1.20(0.90)&0.00(0.00)&1.02(0.01)&1.43(0.94)&0.00(0.00)&1.02(0.01)\\
EBIC&0.05(0.02)&0.00(0.00)&1.02(0.01)&0.06(0.02)&0.00(0.00)&1.02(0.01)\\
\hline
\end{tabular}
}
\end{table}

\begin{table}[!h]
\caption{Comparisons for Example 2 with $\rho = 0$ and $\rho = 0.5$ cases. Results are reported in the form of mean (standard error). For CV($n_v$), $n_c  = \lceil n^{1/2} \rceil$ and $K = 50$; FP, false positive; FN, false negative;  PE, prediction error. \label{tb:ex2}}
\begin{center}
\begin{tabular}{|lrrrrrr|}
\hline
Method & \multicolumn{3}{c}{$\rho = 0$}& \multicolumn{3}{c|}{$\rho = 0.5$} \\ \hline
Lasso &  FP & FN & PE  & FP & FN & PE \\ \hline
CV($n_v$)&0.02(0.01)&0.01(0.01)&1.01(0.01)&0.03(0.02)&0.07(0.03)&1.02(0.01)\\
10-fold&73.56(5.13)&0.00(0.00)&1.17(0.01)&32.47(3.33)&0.00(0.00)&1.09(0.01)\\
1SE&7.44(0.83)&0.00(0.00)&1.23(0.01)&0.99(0.36)&0.00(0.00)&1.15(0.01)\\
AIC&484.59(1.39)&0.00(0.00)&1.41(0.01)&402.84(1.19)&0.00(0.00)&1.31(0.01)\\
BIC&3.41(0.26)&0.00(0.00)&1.24(0.01)&1.10(0.14)&0.00(0.00)&1.11(0.01)\\
EBIC&0.76(0.10)&0.00(0.00)&1.28(0.01)&0.26(0.05)&0.00(0.00)&1.12(0.01)\\
\hline
SCAD & FP & FN & PE  & FP & FN & PE \\ \hline
CV($n_v$)&0.01(0.01)&0.01(0.01)&1.01(0.01)&0.04(0.02)&0.07(0.03)&1.02(0.01)\\
10-fold&19.80(2.35)&0.00(0.00)&1.02(0.01)&22.99(1.78)&0.00(0.00)&1.03(0.01)\\
1SE&0.20(0.06)&0.00(0.00)&1.08(0.01)&1.13(0.23)&0.02(0.01)&1.07(0.01)\\
AIC&214.58(1.31)&0.00(0.00)&1.49(0.02)&34.73(2.48)&0.00(0.00)&1.03(0.01)\\
BIC&0.82(0.11)&0.00(0.00)&1.04(0.01)&1.05(0.16)&0.01(0.01)&1.05(0.01)\\
EBIC&0.19(0.04)&0.00(0.00)&1.04(0.01)&0.30(0.06)&0.02(0.01)&1.06(0.01)\\

\hline
MCP &  FP & FN & PE  & FP & FN & PE \\ \hline

CV($n_v$)&0.01(0.01)&0.01(0.01)&1.01(0.01)&0.04(0.02)&0.07(0.03)&1.02(0.01)\\
10-fold&6.46(1.01)&0.00(0.00)&1.02(0.01)&7.16(0.76)&0.00(0.00)&1.03(0.01)\\
1SE&0.01(0.01)&0.00(0.00)&1.07(0.01)&0.05(0.03)&0.06(0.02)&1.07(0.01)\\
AIC&102.30(0.48)&0.00(0.00)&1.81(0.02)&46.67(1.04)&0.00(0.00)&1.06(0.01)\\
BIC&100.72(1.13)&0.00(0.00)&1.80(0.02)&0.66(0.19)&0.01(0.01)&1.03(0.01)\\
EBIC&0.06(0.03)&0.00(0.00)&1.02(0.01)&0.07(0.03)&0.02(0.01)&1.03(0.01)\\
\hline
\end{tabular}
\end{center}
\end{table}

In Table \ref{T1}, for Lasso penalty, AIC and 10-fold CV have the largest mean FP followed by BIC, 1SE and EBIC. CV($n_v$) performs the best in terms of FP, FN as well as PE in both $\rho=0$ and $\rho=0.5$ cases. 

SCAD and MCP lead to similar performances to Lasso based methods.  Their FP of 10-fold CV are not as large as those of Lasso, but CV($n_v$) still outperforms 10-fold CV and AIC	 in terms of both variable selection and prediction. It is worth to point out that the difference is not as significant as that in the Lasso case, possibly due to the asymptotic unbiasedness property of SCAD and MCP \citep{Zhang2010}.  It is similar that when using BIC and EBIC, SCAD and MCP perform better than Lasso. 

We  also present the comparisons of the $\lambda$ value derived from the universal thresholding \citep{DonohoJohnstone1994} $\lambda_{\mathrm{univ}} = \sigma \sqrt{2(\log p) /n}$ with $\sigma$ being the error standard deviation and the $\lambda$ values from different methods  in Table~\ref{T-compare} under the uncorrelated design ($\rho = 0$). The rationale of universal thresholding is a theoretical upper bound of the maximum of all the errors; hence it can be regarded as a theoretical lower bound of $\lambda$ to remove all the noise variables.  We observe from the table that only CV($n_v$) gives a $\lambda$ value larger than $\lambda_{\mathrm{univ}}$. On the other hand, note that the lowest signal level of this example is 0.4, which can serve as an upper bound of $\lambda$ in order to retain all the important variables.   This analysis provides an explanation to the great performance of CV($n_v$). 

We consider an additional simulation setting described in the following example.
\begin{example}\label{ex::CCV2} Linear regression. For $i = 1, \cdots, n$, let $y_i = \boldsymbol{x}^{\top}_i\boldsymbol{\beta}^o + \varepsilon_i$, where $\boldsymbol{x}_i \stackrel{i.i.d.}{\sim} \mathcal{N}(\boldsymbol{0}_p, \Sigma)$ with $\boldsymbol{0}_p$ the length-$p$ vector with all 0 entries and $\Sigma_{j,k} = \rho^{|j-k|}$, $\varepsilon_i \stackrel{i.i.d.}{\sim} \mathcal{N}(0, 1)$, $\rho = 0$ or $0.5$, $(n, p) = (500, 10000)$ and $\boldsymbol{\beta}^o \in \mathbb{R}^p$ with the first 7 coordinates (1, 0.9, 0.8, 0.7, 0.6, 0.5, 0.4) and 0 elsewhere. 
\end{example}
Note that this is a more challenging scenario compared with Example 1(i) since there are more signal variables and the correlation among signal variables is stronger when $\rho=0.5$. 
The corresponding results for Example \ref{ex::CCV2} are available in Table \ref{tb:ex2}. For $\rho=0$ with Lasso penalty, CV($n_v$) performs significantly better than all the competing methods in terms of both FP and PE. We observe similar conclusions for other settings. 


\begin{table}
\caption{Comparison of $\lambda$ values derived from various methods for Example 2 (i) with $\rho=0$. Results are presented in the form of mean (standard error).}
\label{T-compare} \par
\vskip .2cm
\centerline{\tabcolsep=5truept\begin{tabular}{|ccccccc|} \hline
 Universal & CV($n_v$) & 10-fold & 1SE & AIC &BIC&EBIC \\ \hline
0.19&0.20(0.02)&0.12(0.02)&0.18(0.02)&0.01(0.00)&0.17(0.01)&0.18(0.01)\\\hline
\end{tabular}}
\end{table}

\noindent {\bf 4.3. Logistic regression}\\
For logistic regression, we use the  setting in Example~\ref{lam-sync} (ii) with $\rho = 0$ and $\rho = 0.5$, and repeat the simulation 100 times. In Table~\ref{T2}, for CV($n_v$), we set $n_c = \lceil n^{3/4} \rceil = 106$ following the results in Section 3.2. Different from the linear case, instead of reporting PE, we report classification error (CE),  which is defined as the average classification error evaluated at an independent test data set of size $n$. The remaining settings and packages used are the same as in the linear regression case. 

\begin{table}[t]

\caption{Comparison in Logistic regression with $\rho = 0$ and $\rho = 0.5$ cases. Results are reported in the form of mean (standard error). For CV($n_v$), $n_c  = \lceil n^{3/4} \rceil$ and $K = 50$;  FN, false negative; FP, false positive; CE, classification error.}
\label{T2}\par
\vskip .2cm
\centerline{\tabcolsep=5truept
\begin{tabular}{|lrrrrrr|}\hline
Method & \multicolumn{3}{c}{$\rho = 0$}& \multicolumn{3}{c|}{$\rho = 0.5$} \\ \hline
Lasso &  FP & FN & CE(\%)  & FP & FN & CE(\%) \\ \hline
CV($n_v$)&1.63(0.14)&0.01(0.01)&19.34(0.20)&0.92(0.10)&0.10(0.03)&16.06(0.20)\\
10-fold&95.86(4.60)&0.00(0.00)&20.81(0.23)&87.76(4.28)&0.01(0.01)&17.22(0.20)\\
1SE&21.48(2.05)&0.00(0.00)&19.44(0.20)&15.60(1.67)&0.03(0.02)&16.19(0.18)\\
AIC&21.63(1.34)&0.00(0.00)&19.50(0.20)&20.01(1.56)&0.02(0.01)&16.28(0.19)\\
BIC&1.88(0.14)&0.08(0.03)&19.49(0.19)&1.75(0.16)&0.05(0.02)&16.14(0.18)\\
EBIC&0.46(0.07)&0.16(0.04)&19.72(0.20)&0.60(0.08)&0.11(0.03)&16.25(0.18)\\
\hline
SCAD & FP & FN & CE(\%)   & FP & FN & CE(\%)  \\ \hline
CV($n_v$)&1.84(0.13)&0.02(0.01)&19.48(0.22)&1.17(0.12)&0.08(0.03)&16.21(0.19)\\
10-fold&55.05(2.06)&0.00(0.00)&19.22(0.20)&52.40(1.93)&0.02(0.01)&16.72(0.19)\\
1SE&10.88(0.76)&0.00(0.00)&19.34(0.19)&8.29(0.70)&0.03(0.02)&16.40(0.18)\\
AIC&31.24(1.25)&0.00(0.00)&19.20(0.20)&23.84(1.34)&0.06(0.02)&16.48(0.19)\\
BIC&3.23(0.26)&0.03(0.02)&19.61(0.19)&2.17(0.20)&0.08(0.03)&16.41(0.19)\\
EBIC&0.92(0.10)&0.11(0.03)&19.90(0.19)&0.77(0.08)&0.10(0.03)&16.54(0.19)\\
\hline
MCP &  FP & FN & CE(\%)  & FP & FN & CE(\%) \\ \hline
CV($n_v$)&2.08(0.12)&0.02(0.01)&19.76(0.20)&1.36(0.10)&0.06(0.02)&16.60(0.19)\\
10-fold&13.10(0.79)&0.00(0.00)&18.84(0.21)&13.31(0.87)&0.03(0.02)&16.23(0.19)\\
1SE&0.91(0.14)&0.06(0.02)&19.13(0.19)&1.09(0.22)&0.10(0.03)&16.25(0.19)\\
AIC&19.38(1.08)&0.00(0.00)&18.92(0.20)&33.39(1.02)&0.03(0.02)&16.86(0.20)\\
BIC&2.51(0.23)&0.03(0.02)&18.88(0.20)&2.16(0.27)&0.08(0.03)&16.06(0.18)\\
EBIC&0.49(0.07)&0.07(0.03)&18.97(0.19)&0.32(0.06)&0.13(0.03)&16.16(0.18)\\
\hline
\end{tabular}
}\end{table}

In Table~\ref{T2}, CV($n_v$) significantly outperforms 10-fold CV and  AIC in terms of FP. The difference is more significant than that in the linear regression case, when SCAD or MCP is used. For Lasso penalized logistic regression, 1SE has a significant number of FP, compared with the linear regression case. EBIC tends to work much better than AIC and BIC, with a similar performance as CV($n_v$) when SCAD and MCP are applied.  When evaluated by the CE, CV($n_v$) still performs the best in most scenarios.

\noindent {\bf 5. Data Analysis}

We now illustrate two applications of the proposed method via eye disease gene expression data \citep{Scheetz2006} and leukemia data set \citep{Golub}. 

In the eye disease gene expression data set, for harvesting of tissue from the eyes and subsequent microarray analysis, 120 12-week-old male rats were selected. The microarrays used to analyze the RNA from the eyes of these animals contain more than 31,042 different probe sets (Affymetric GeneChip Rat Genome 230 2.0 Array). The intensity values were normalized using the robust multichip averaging method \citep{Irizarry2003} to obtain summary expression values for each probe set. Gene expression levels were analyzed on a logarithmic scale.

Following \cite{HuangEtal2010} and \cite{FanEtal2011}, we are interested in finding the genes that are related to the TRIM32 gene, which was recently found to cause Bardet--Biedl syndrome \citep{Chiang2006} and is a genetically heterogeneous disease of multiple organ systems, including the retina. Although more than 30,000 probe sets are represented on the Rat Genome 230 2.0 Array, many of these are not expressed in the eye tissue. We focus only on the 18,975 probes that are expressed in the eye tissue.

The leukemia data set we studied was previously analyzed in \cite{Golub}. There are $p = 7,129$ genes and $n = 72$ samples coming from two classes: 47 in class ALL (acute lymphocytic leukemia) and 25 in class AML (acute myelogenous leukemia).  

We model these two problems using linear and logistic regression, respectively.  In the eye gene expression data set, we randomly draw without replacement 100 out of 120 observations from the sample, using them as training data, and use the remaining sub-sample of size 20 as the test data.  In the leukemia data set, we randomly draw without replacement  60 out of 72 observations from the sample as the training data with the remaining observations as the test data. 

We repeat this procedure 100 times with the results reported in Tables  \ref{tb::rat} and \ref{tb::leukemia} in the form of mean (standard error). For each split, we use \pkg{glmnet} and \pkg{ncvreg} to compute the Lasso and SCAD / MCP collections of solutions, respectively; we then compare our proposed CV($n_v$) with the 10-fold CV, which is the default tuning parameter selection method in \pkg{glmnet} and \pkg{ncvreg}. In addition, we investigate the performance of 1SE, AIC,  BIC and  EBIC.

\begin{table}[t]
\caption{Model size and prediction error for the Eye Disease Gene Expression data sets. Results are reported in the form of mean (standard error). \label{tb::rat}}
\begin{center}
\begin{tabular}{|l|rr|rr|rr|}
\hline
&\multicolumn{2}{c|}{Lasso}&\multicolumn{2}{c|}{SCAD}&\multicolumn{2}{c|}{MCP}\\\hline 
Method&Size&PE&Size&PE&Size&PE\\
\hline
CV($n_v$)&2.46(0.08)&0.01(0.00)&2.23(0.07)&0.01(0.00)&2.36(0.07)&0.01(0.00)\\
10-fold&61.18(1.68)&0.01(0.00)&33.54(0.59)&0.01(0.00)&11.12(0.30)&0.01(0.00)\\
1SE&31.03(1.16)&0.01(0.00)&24.84(0.71)&0.01(0.00)&5.39(0.31)&0.01(0.00)\\
AIC&103.02(0.48)&0.01(0.00)&0.37(0.05)&0.02(0.00)&5.38(0.25)&0.01(0.00)\\
BIC&99.99(0.71)&0.01(0.00)&0.17(0.04)&0.02(0.00)&4.65(0.25)&0.01(0.00)\\
EBIC&1.03(0.24)&0.02(0.00)&0.02(0.01)&0.02(0.00)&1.90(0.13)&0.01(0.00)\\
\hline
\end{tabular}
\end{center}
\end{table}

\begin{table}[t]
\caption{Model size and test classification error for the Leukemia data sets. Results are reported in the form of mean (standard error). \label{tb::leukemia}}
\begin{center}
\begin{tabular}{|l|rr|rr|rr|}
\hline
&\multicolumn{2}{c|}{LASSO}&\multicolumn{2}{c|}{SCAD}&\multicolumn{2}{c|}{MCP}\\\hline 
Method&Size&CE(\%)&Size&CE(\%)&Size&CE(\%)\\
\hline
CV($n_v$)&8.93(0.58)&8.07(0.75)&10.14(0.57)&7.80(0.79)&5.62(0.14)&8.33(0.76)\\
10-fold&21.52(0.41)&5.85(0.74)&17.04(0.31)&7.45(0.75)&5.07(0.14)&9.84(0.94)\\
1SE&12.51(0.46)&8.87(0.93)&11.70(0.43)&9.84(0.92)&2.85(0.16)&14.54(1.25)\\
AIC&16.17(0.39)&6.65(0.76)&1.00(0.05)&30.59(1.37)&4.14(0.17)&10.02(0.98)\\
BIC&4.32(0.30)&17.20(1.14)&0.91(0.06)&30.59(1.37)&3.56(0.15)&10.37(1.01)\\
EBIC&0.48(0.06)&31.29(1.33)&0.37(0.05)&31.56(1.32)&1.46(0.08)&14.72(1.03)\\
\hline
\end{tabular}
\end{center}
\end{table}

For the eye disease gene expression data sets, Table \ref{tb::rat} shows that CV($n_v$) performs well compared with 10-fold CV as well as all the considered information type criteria. In terms of model size, EBIC leads to the smallest model on average when using the Lasso penalty.  It, however, probably misses some important predictors as the prediction error is larger than those of the other methods. Among the models that give the best prediction error, CV($n_v$) always selects the sparsest model. Similar behaviors can be found in SCAD and MCP, though the differences in performances are not as pronounced. 

For the leukemia data set, we can see from Table \ref{tb::leukemia} that both BIC and EBIC select very small models with a  large test classification error for all three penalties. CV($n_v$) tends to provide a reasonably good balance for the complexity of the model and the test classification error. Although 10-fold CV has a smaller test classification error for Lasso and SCAD, it selects many more variables on average.

\noindent {\bf 6. Discussion}

In this paper we study CV methods applied to the tuning parameter selection problem in high-dimensional penalized generalized linear models.  For the $K$-fold CV, we show the issue of the mis-alignment for different splits is one possible reason of over-selection. We advocate the use of CV($n_v$) with a proper choice of $n_v$ for the path algorithms, which has been shown to be restricted model selection consistent in high-dimensional settings.

One possible future direction is to study the theoretical implication of low coherent rate of CV, as demonstrated in the numerical results,  on the model selection performance. The proposed algorithm is a general framework, which could be extended to the case of using other methods (e.g., forward regression) to generate the collection of solutions. It is also interesting to extend the methodology and the associated theory  to other models including additive models, Cox proportional hazards models, among others.  In addition, we are interested in selecting the concavity parameter $\gamma$ in folded-concave penalties via cross-validation.

An implementation of the CV($n_v$) method for high-dimensional variable selection is available at \url{https://github.com/statcodes/rccv}.

\vskip 14pt
\noindent {\large\bf Supplementary Materials}

The online supplementary materials include all the technical details and additional simulation results.

\vskip 14pt
\noindent {\large\bf Acknowledgements}

The authors would like to thank the co-Editor Professor Hsin-Cheng Huang, the AE and three referees for their insightful comments, which have greatly improved the scope and quality of the paper. 
This research was partially supported by NSF CAREER grant DMS-1554804.

\bibliographystyle{ims}
\bibliography{GLMCV}

\begin{thebibliography}{34}
\expandafter\ifx\csname natexlab\endcsname\relax\def\natexlab#1{#1}\fi
\expandafter\ifx\csname url\endcsname\relax
  \def\url#1{\texttt{#1}}\fi
\expandafter\ifx\csname urlprefix\endcsname\relax\def\urlprefix{URL }\fi
\providecommand{\eprint}[2][]{\url{#2}}

\bibitem[{Breheny and Huang(2011)}]{BrehenyHuang2011}
\textsc{Breheny, P.} and \textsc{Huang, J.} (2011).
\newblock Coordinate descent algorithms for nonconvex penalized regression,
  with applications to biological feature selection.
\newblock \textit{The Annals of Applied Statistics}, \textbf{5} 232--253.

\bibitem[{B\"uhlmann and van~de Geer(2011)}]{BuehlmannGeer2011}
\textsc{B\"uhlmann, P.} and \textsc{van~de Geer, S.} (2011).
\newblock \textit{Statistics for High-Dimensional Data}.
\newblock Springer-Verlag New York Inc.

\bibitem[{Chen and Chen(2008)}]{ChenChen2008}
\textsc{Chen, J.} and \textsc{Chen, Z.} (2008).
\newblock Extended bayesian information criteria for model selection with large
  model spaces.
\newblock \textit{Biometrika}, \textbf{95} 759--71.

\bibitem[{Chen and Donoho(1994)}]{ChenDonoho1994}
\textsc{Chen, S.} and \textsc{Donoho, D.} (1994).
\newblock Basis pursuit.
\newblock In \textit{1994 Conference Record of the Twenty-Eighth Asilomar
  Conference on Signals, Systems and Computers}, vol.~1. IEEE, 41--4.

\bibitem[{Chiang et~al.(2006)Chiang, Beck, Yen, Tayeh, Scheetz, Swiderski,
  Nishimura, Braun, Kim, Huang, Elbedour, Carmi, Slusarski, Casavant, Stone and
  Sheffield}]{Chiang2006}
\textsc{Chiang, A.~P.}, \textsc{Beck, J.~S.}, \textsc{Yen, H.-J.},
  \textsc{Tayeh, M.~K.}, \textsc{Scheetz, T.~E.}, \textsc{Swiderski, R.~E.},
  \textsc{Nishimura, D.~Y.}, \textsc{Braun, T.~A.}, \textsc{Kim, K.-Y.~A.},
  \textsc{Huang, J.}, \textsc{Elbedour, K.}, \textsc{Carmi, R.},
  \textsc{Slusarski, D.~C.}, \textsc{Casavant, T.~L.}, \textsc{Stone, E.~M.}
  and \textsc{Sheffield, V.~C.} (2006).
\newblock Homozygosity mapping with snp arrays identifies trim32, an e3
  ubiquitin ligase, as a bardet--biedl syndrome gene (bbs11).
\newblock \textit{Proc. Natl. Acad. Sci. U.S.A.}, \textbf{103} 6287--92.

\bibitem[{Donoho and Johnstone(1994)}]{DonohoJohnstone1994}
\textsc{Donoho, D.~L.} and \textsc{Johnstone, I.~M.} (1994).
\newblock Ideal spatial adaptation by wavelet shrinkage.
\newblock \textit{Biometrika}, \textbf{81} 425--455.

\bibitem[{Efron(1983)}]{Efron1983}
\textsc{Efron, B.} (1983).
\newblock Estimating the error rate of a prediction rule: Improvement on
  cross-validation.
\newblock \textit{Biometrika}, \textbf{78} 316--31.

\bibitem[{Efron(1986)}]{Efron1986}
\textsc{Efron, B.} (1986).
\newblock How biased is the apparent error rate of a prediction rule?
\newblock \textit{J. Am. Statist. Assoc.}, \textbf{81} 461--70.

\bibitem[{Efron et~al.(2004)Efron, Hastie, Johnstone, Tibshirani
  et~al.}]{efron2004least}
\textsc{Efron, B.}, \textsc{Hastie, T.}, \textsc{Johnstone, I.},
  \textsc{Tibshirani, R.} \textsc{et~al.} (2004).
\newblock Least angle regression.
\newblock \textit{The Annals of statistics}, \textbf{32} 407--499.

\bibitem[{Fan et~al.(2011)Fan, Feng and Song}]{FanEtal2011}
\textsc{Fan, J.}, \textsc{Feng, Y.} and \textsc{Song, R.} (2011).
\newblock Nonparametric independence screening in sparse ultra-high-dimensional
  additive models.
\newblock \textit{J. Am. Statist. Assoc.}, \textbf{106} 544--57.

\bibitem[{Fan and Li(2001)}]{FanLi2001}
\textsc{Fan, J.} and \textsc{Li, R.} (2001).
\newblock Variable selection via nonconcave penalized likelihood and its oracle
  properties.
\newblock \textit{J. Am. Statist. Assoc.}, \textbf{96} 1348--60.

\bibitem[{Fan and Lv(2010)}]{FanLv2010}
\textsc{Fan, J.} and \textsc{Lv, J.} (2010).
\newblock A selective overview of variable selection in high dimensional
  feature space.
\newblock \textit{Statistica Sinica}, \textbf{20} 101--148.

\bibitem[{Fan and Tang(2013)}]{Fan.Tang.2012}
\textsc{Fan, Y.} and \textsc{Tang, C.~Y.} (2013).
\newblock Tuning parameter selection in high-dimensional penalized likelihood.
\newblock \textit{Journal of the Royal Statistical Society, Ser. B.},
  \textbf{75} 531--552.

\bibitem[{Friedman et~al.(2010)}]{FriedmanHT2010}
\textsc{Friedman, J.} \textsc{et~al.} (2010).
\newblock Regularization paths for generalized linear models via coordinate
  descent.
\newblock \textit{J. Statist. Softw.}, \textbf{33} 1--22.

\bibitem[{Golub et~al.(1999)Golub, Slonim, Tamayo, Huard, Gaasenbeek, Mesirov,
  Coller, Loh, Downing, Caligiuri, Bloomfield and Lander}]{Golub}
\textsc{Golub, T.~R.}, \textsc{Slonim, D.~K.}, \textsc{Tamayo, P.},
  \textsc{Huard, C.}, \textsc{Gaasenbeek, M.}, \textsc{Mesirov, J.~P.},
  \textsc{Coller, H.}, \textsc{Loh, M.~L.}, \textsc{Downing, J.~R.},
  \textsc{Caligiuri, M.~A.}, \textsc{Bloomfield, C.~D.} and \textsc{Lander,
  E.~S.} (1999).
\newblock Molecular classification of cancer: class discovery and class
  prediction by gene expression monitoring.
\newblock \textit{Science}, \textbf{286} 531--37.

\bibitem[{Huang et~al.(2010)Huang, Horowitz and Wei}]{HuangEtal2010}
\textsc{Huang, J.}, \textsc{Horowitz, J.~L.} and \textsc{Wei, F.} (2010).
\newblock Variable selection in nonparametric additive models.
\newblock \textit{Ann. Statist.}, \textbf{38} 2282--313.

\bibitem[{Irizarry et~al.(2003)Irizarry, Hobbs, Collin, Beazer-Barclay,
  Antonellis, Scherf and Speed}]{Irizarry2003}
\textsc{Irizarry, R.~A.}, \textsc{Hobbs, B.}, \textsc{Collin, F.},
  \textsc{Beazer-Barclay, Y.~D.}, \textsc{Antonellis, K.~J.}, \textsc{Scherf,
  U.} and \textsc{Speed, T.~P.} (2003).
\newblock Exploration, normalization, and summaries of high density
  oligonucleotide array probe level data.
\newblock \textit{Biostatistics}, \textbf{4} 249--64.

\bibitem[{Luo and Chen(2014)}]{luo2014sequential}
\textsc{Luo, S.} and \textsc{Chen, Z.} (2014).
\newblock Sequential lasso cum ebic for feature selection with ultra-high
  dimensional feature space.
\newblock \textit{Journal of the American Statistical Association},
  \textbf{109} 1229--1240.

\bibitem[{Meinshausen(2007)}]{Meinshausen}
\textsc{Meinshausen, N.} (2007).
\newblock Relaxed lasso.
\newblock \textit{Comput. Stat. Data Anal.}, \textbf{52} 374--93.

\bibitem[{Park and Hastie(2007)}]{ParkHastie2007}
\textsc{Park, M.~Y.} and \textsc{Hastie, T.} (2007).
\newblock An $l_1$ regularization-path algorithm for generalized linear models.
\newblock \textit{J. R. Statist. Soc. B}, \textbf{69} 659--677.

\bibitem[{Scheetz et~al.(2006)Scheetz, Kim, Swiderski, Philp, Braun, Knudtson,
  Dorrance, DiBona, Huang, Casavant, Sheffield and Stone}]{Scheetz2006}
\textsc{Scheetz, T.~E.}, \textsc{Kim, K.-Y.~A.}, \textsc{Swiderski, R.~E.},
  \textsc{Philp, A.~R.}, \textsc{Braun, T.~A.}, \textsc{Knudtson, K.~L.},
  \textsc{Dorrance, A.~M.}, \textsc{DiBona, G.~F.}, \textsc{Huang, J.},
  \textsc{Casavant, T.~L.}, \textsc{Sheffield, V.~C.} and \textsc{Stone, E.~M.}
  (2006).
\newblock Regulation of gene expression in the mammalian eye and its relevance
  to eye disease.
\newblock \textit{Proc. Natl. Acad. Sci. U.S.A.}, \textbf{103} 14429--34.

\bibitem[{Shao(1993)}]{Shao1993}
\textsc{Shao, J.} (1993).
\newblock Linear model selection by cross-validation.
\newblock \textit{J. Am. Statist. Assoc.}, \textbf{88} 486--94.

\bibitem[{Shao(1996)}]{Shao1996}
\textsc{Shao, J.} (1996).
\newblock Bootstrap model selection.
\newblock \textit{J. Am. Statist. Assoc.}, \textbf{91} 655--65.

\bibitem[{St{\"a}dler et~al.(2010)St{\"a}dler, B{\"u}hlmann and Van
  De~Geer}]{Stadler.Bulmann.ea.2010}
\textsc{St{\"a}dler, N.}, \textsc{B{\"u}hlmann, P.} and \textsc{Van De~Geer,
  S.} (2010).
\newblock l1-penalization for mixture regression models.
\newblock \textit{Test}, \textbf{19} 209--256.

\bibitem[{Stone(1977)}]{Stone1977}
\textsc{Stone, M.} (1977).
\newblock An asymptotic equivalence of choice of model by cross-validation and
  akaike's criterion.
\newblock \textit{J. R. Statist. Soc. B}, \textbf{39} 44--7.

\bibitem[{Sun and Zhang(2012)}]{Sun.Zhang.2011}
\textsc{Sun, T.} and \textsc{Zhang, C.-H.} (2012).
\newblock Scaled sparse linear regression.
\newblock \textit{Biometrika}, \textbf{99} 879--898.

\bibitem[{Tibshirani(1996)}]{Tibshirani1996}
\textsc{Tibshirani, R.} (1996).
\newblock Regression shrinkage and selection via the lasso.
\newblock \textit{J. R. Statist. Soc. B}, \textbf{9} 267--88.

\bibitem[{Wang et~al.(2009)Wang, Li and Leng}]{wang2009shrinkage}
\textsc{Wang, H.}, \textsc{Li, B.} and \textsc{Leng, C.} (2009).
\newblock Shrinkage tuning parameter selection with a diverging number of
  parameters.
\newblock \textit{Journal of the Royal Statistical Society: Series B
  (Statistical Methodology)}, \textbf{71} 671--683.

\bibitem[{Wang et~al.(2007)Wang, Li and Tsai}]{wang2007tuning}
\textsc{Wang, H.}, \textsc{Li, R.} and \textsc{Tsai, C.-L.} (2007).
\newblock Tuning parameter selectors for the smoothly clipped absolute
  deviation method.
\newblock \textit{Biometrika}, \textbf{94} 553--68.

\bibitem[{Yu and Feng(2014{\natexlab{a}})}]{yu2014apple}
\textsc{Yu, Y.} and \textsc{Feng, Y.} (2014{\natexlab{a}}).
\newblock Apple: Approximate path for penalized likelihood estimators.
\newblock \textit{Statistics and Computing}, \textbf{24} 803--819.

\bibitem[{Yu and Feng(2014{\natexlab{b}})}]{yu2014modified}
\textsc{Yu, Y.} and \textsc{Feng, Y.} (2014{\natexlab{b}}).
\newblock Modified cross-validation for penalized high-dimensional linear
  regression models.
\newblock \textit{Journal of Computational and Graphical Statistics},
  \textbf{23} 1009--1027.

\bibitem[{Zhang(2010)}]{Zhang2010}
\textsc{Zhang, C.-H.} (2010).
\newblock Nearly unbiased variable selection under minimax concave penalty.
\newblock \textit{Ann. Statist.}, \textbf{38} 894--942.

\bibitem[{Zhang and Huang(2008)}]{ZhangHuang2008}
\textsc{Zhang, C.-H.} and \textsc{Huang, J.} (2008).
\newblock The sparsity and bias of the lasso selection in high-dimensional
  linear regression.
\newblock \textit{Ann. Statist.}, \textbf{36} 1567--94.

\bibitem[{Zhang et~al.(2010)Zhang, Li and Tsai}]{ZhangEtal2010}
\textsc{Zhang, Y.}, \textsc{Li, R.} and \textsc{Tsai, C.-L.} (2010).
\newblock Regularization parameter selections via generalized information
  criterion.
\newblock \textit{J. Am. Statist. Assoc.}, \textbf{105} 312--23.

\end{thebibliography}


\vskip .65cm
\noindent
Department of Statistics, Columbia University
\vskip 2pt
\noindent
E-mail: yang.feng@columbia.edu
\vskip 2pt

\noindent
School of Mathematics, University of Bristol
\vskip 2pt
\noindent
E-mail: y.yu@bristol.ac.uk

\newpage
\fontsize{10.95}{14pt plus.8pt minus .6pt}\selectfont
\vspace{0.8pc}
\centerline{\large\bf Supplementary materials for ``The restricted consistency }
\vspace{2pt}
\centerline{\large\bf  property of leave-$n_v$-out cross-validation for }
\vspace{2pt}
\centerline{\large\bf   high-dimensional variable selection"}
\vspace{.4cm}
\centerline{Yang Feng and Yi Yu}
\vspace{.4cm}
\centerline{\it Columbia University and University of Bristol}
\vspace{.55cm}
\fontsize{9}{11.5pt plus.8pt minus .6pt}\selectfont


The supplementary material includes all the technical details and additional simulation results.

\noindent {\bf A.1 Additional lemmas and proofs}

The following Lemma is adapted from \cite{Lalley2013}.  It helps us to develop the asymptotic theory where $N$, the size of the candidate models, is allowed to diverge with the sample size. 

\begin{lemma}[Gaussian concentration]\label{lemma-tala}
Let $\gamma$ be the standard Gaussian probability measure on $\mathbb{R}^n$ (that is, the distribution of a $\mathcal{N}(0, I_n)$ random vector), and let $F: \mathbb{R}^n \to \mathbb{R}$ be Lipschitz in each variable separately relative to the Euclidean metric, with Lipschitz constant $c$.  Then for every $t > 0$,
\begin{align*}
\gamma\{|F - E_{\gamma}(F)| \ge t\} \le 2\exp\left(-\frac{t^2}{c^2\pi^2}\right).
\end{align*}
\end{lemma}

\begin{lemma}\label{classic}
With $p<n$, let $\tilde{\beta}$ be the MLEs of a generalized linear model.  Assume the  penalty function $p(\cdot)$ is separable, and assume Conditions~1 - 6 hold. Furthermore, assume $n_c \to \infty$ and $n_c/n \to 0$ as $n\to \infty$, and the size of the splits $K$ satisfies
\begin{align*}
K^{-1}n_c^{-2}n^2 \to 0.
\end{align*}
Then, CV($n_v$) with $K$ times subsampling is restricted model selection consistent.
\end{lemma}

\begin{proof}[Proof of Lemma~\ref{classic}]

Due to the properties of generalized linear models with canonical parameter, we have
\begin{align*}
E(y_i \mid \boldsymbol{x}_i) = \dot b(\boldsymbol{x}_i^{\top}\boldsymbol{\beta}), \quad \sigma_i^2 = a(\phi) \ddot b(\boldsymbol{x}_i^{\top}\boldsymbol{\beta}), \quad i = 1,\cdots, n,
\end{align*}
and define $\sigma^2 = (1/n)\sum_{i=1}^n \sigma_i^2$.
The target is to select the model that minimizes the loss
\begin{align}\label{e-pf1-1}
{\tilde \Gamma}_{\alpha} = \frac{1}{K n_v} \sum_{s\in \mathcal{S}} \Bigl\{-\boldsymbol{y}_s^{\top}(X_s^{\alpha}\tilde{\boldsymbol{\beta}}_{s^c,\alpha}) + \boldsymbol{1}^{\top}b(X_s^{\alpha}\tilde{\boldsymbol{\beta}}_{s^c,\alpha})\Bigr\},
\end{align}
where $\mathcal{S}$ represents the collection of validation sets in different splits and $\boldsymbol{1}$ is an all-one vector.

Denote $E_{\mathcal{S}}$ and $\mathrm{var}_{\mathcal{S}}$ as the expectation and variance with respect to the random selection of $\mathcal{S}$.  By using the equality
\begin{align*}
E_{\mathcal{S}}\biggl(\frac{1}{r}\sum_{s\in \mathcal{S}} a_s \biggr) = \binom{n}{n_v}^{-1}\sum_{s \in \text{ all } s}E(a_s),
\end{align*}
rewriting (\ref{e-pf1-1}), and denoting $\ell_s (\boldsymbol{\beta}) = \boldsymbol{y}_s^{\top}(X_s\boldsymbol{\beta}) - \boldsymbol{1}^{\top}b(X_s\boldsymbol{\beta})$ and $\ell_n (\tilde{\boldsymbol{\beta}}_{\alpha}) = \boldsymbol{y}^{\top}(X_{\alpha}\tilde{\boldsymbol{\beta}}_{\alpha}) - \boldsymbol{1}^{\top}b(X_{\alpha}\tilde{\boldsymbol{\beta}}_{\alpha})$, we have
\begin{align*}
E_{\mathcal{S}}({\tilde \Gamma}_{\alpha})  = & E_{\mathcal{S}} \biggl(-\frac{1}{K n_v} \sum_{s\in \mathcal{S}} \ell_s(\boldsymbol{\beta}^o)\biggr) +  E_{\mathcal{S}} \biggl(\frac{1}{K n_v} \sum_{s\in \mathcal{S}}\bigl\{\ell_s(\boldsymbol{\beta}^o) - \bigl(\boldsymbol{y}^{\top}_s(X_s^{\alpha}\tilde{\boldsymbol{\beta}}_{\alpha}) - \boldsymbol{1}^{\top}b(X_s^{\alpha}\tilde{\boldsymbol{\beta}}_{\alpha})\bigr)\bigr\}\bigg) \\
+ & E_{\mathcal{S}} \biggl( \frac{1}{K n_v}\sum_{s\in \mathcal{S}}\bigl\{\bigl(\boldsymbol{y}^{\top}_s(X_s^{\alpha}\tilde{\boldsymbol{\beta}}_{\alpha}) - \boldsymbol{1}^{\top}b(X_s^{\alpha}\tilde{\boldsymbol{\beta}}_{\alpha})\bigr) - \bigl(\boldsymbol{y}^{\top}_s(X_s^{\alpha}\tilde{\boldsymbol{\beta}}_{s^c, \alpha}) - \boldsymbol{1}^{\top}b(X_s^{\alpha}\tilde{\boldsymbol{\beta}}_{s^c, \alpha})\bigr) \bigr\}\biggr) \\
= & E \biggl(-\frac{1}{n} \ell_n(\boldsymbol{\beta}^o) + \frac{1}{n}\bigl(\ell_n(\boldsymbol{\beta}^o) - \ell_n(\tilde{\boldsymbol{\beta}}_{\alpha})\bigr) \\
+ & \binom{n}{n_v}^{-1}\sum_{s \in  \text{ all } s}\frac{1}{n_v}\bigl\{\boldsymbol{y}_s^{\top}(X_s^{\alpha}\tilde{\boldsymbol{\beta}}_{\alpha} - X_s^{\alpha}\tilde{\boldsymbol{\beta}}_{s^c, \alpha}) - \boldsymbol{1}^{\top}\bigl(b(X_s^{\alpha}\tilde{\boldsymbol{\beta}}_{\alpha}) - b(X_s^{\alpha}\tilde{\boldsymbol{\beta}}_{s^c, \alpha})\bigr)\bigr\}\bigg) \\
= & -\frac{1}{n} E (\ell_n(\boldsymbol{\beta}^o)) + E(A_{\alpha 1}) + \binom{n}{n_v}^{-1}\sum_{s \in  \text{ all } s}E(A_{\alpha 2, s}).
\end{align*}

For different $\alpha$, $E(\ell_n(\boldsymbol{\beta}^o))$ stays the same, so we only need to focus on $A_{\alpha 1}$ and $A_{\alpha 2,s}$.

From Wilks' theorem, it is known that, if $\alpha \in \cA\setminus\cA_c$, as $n\to\infty$, we have $A_{\alpha 1}\stackrel{\mathcal{D}}{\to}(1/2) \chi^2(k_{\alpha})$, where $k_{\alpha} = d_0 - d_{\alpha 0}$, $d_{\alpha 0} = |\{j: \beta_j \in \alpha \cap \alpha_0\}|$, i.e., $k_{\alpha}$ is the number of false negatives. This means $E(A_{\alpha 1}) = k_{\alpha}$; otherwise, $E(A_{\alpha 1}) = O(1/n)$.

For any $s$,
\begin{align*}
& \boldsymbol{1}^{\top}\bigl(b(X_s^{\alpha}\tilde{\boldsymbol{\beta}}_{\alpha}) - b(X_s^{\alpha}\tilde{\boldsymbol{\beta}}_{s^c, \alpha})\bigr) = \bigl(\dot b(X_s^{\alpha}\tilde{\boldsymbol{\beta}}_{\alpha})\bigr)^{\top}X_s^{\alpha}(\tilde{\boldsymbol{\beta}}_{\alpha} - \tilde{\boldsymbol{\beta}}_{s^c, \alpha}) \\
& - \frac{1}{2}(\tilde{\boldsymbol{\beta}}_{\alpha} - \tilde{\boldsymbol{\beta}}_{s^c, \alpha})^{\top}(X_s^{\alpha})^{\top}\ddot b(X_s^{\alpha}\tilde{\boldsymbol{\beta}}_{\alpha})X_{s, \alpha}(\tilde{\boldsymbol{\beta}}_{\alpha} - \tilde{\boldsymbol{\beta}}_{s^c, \alpha}) + o(1).
\end{align*}
Define $u_{s^c}(\boldsymbol{\gamma}) = (1/n_c)(X_{s^c}^{\alpha})^{\top}\bigl(\boldsymbol{y}_{s^c} - \dot b(X_{s^c}^{\alpha}\boldsymbol{\gamma})\bigr)$, then $\tilde{\boldsymbol{\beta}}_{s^c, \alpha}$ is the solution to $u_{s^c}(\boldsymbol{\gamma}) = 0$. By Taylor expansion, we get
\[
\tilde{\boldsymbol{\beta}}_{\alpha} - \tilde{\boldsymbol{\beta}}_{s^c, \alpha} = \bigl(\dot u_{s^c}(\tilde{\boldsymbol{\beta}}_{\alpha})\bigr)^{-1}u_{s^c}(\tilde{\boldsymbol{\beta}}_{\alpha})(1+ o(1)),
\]
where $\dot u_{s^c}(\tilde{\boldsymbol{\beta}}_{\alpha}) = -(1/n_c)(X_{s^c}^{\alpha})^{\top}{\ddot b}(X_{s^c}^{\alpha}\tilde{\boldsymbol{\beta}}_{\alpha})X_{s^c}^{\alpha}$.

Define $D_{s, \alpha} = \ddot b^{1/2}(X_s^{\alpha}\tilde{\boldsymbol{\beta}}_{\alpha})X_{s^c}^{\alpha}$, then
\begin{align*}
 A_{\alpha 2, s} &= \frac{1}{n_v} \bigl(\boldsymbol{y}_s - \dot b(X_s^{\alpha}\tilde{\boldsymbol{\beta}}_{\alpha})\bigr)^{\top}X_s^{\alpha}(\tilde{\boldsymbol{\beta}}_{\alpha} - \tilde{\boldsymbol{\beta}}_{s^c, \alpha}) \\
&+ \frac{1}{2 n_v}(\tilde{\boldsymbol{\beta}}_{\alpha} - \tilde{\boldsymbol{\beta}}_{s^c, \alpha})^{\top}(X_s^{\alpha})^{\top}\ddot b(X_s^{\alpha}\tilde{\boldsymbol{\beta}}_{\alpha})X_s^{\alpha}(\tilde{\boldsymbol{\beta}}_{\alpha} - \tilde{\boldsymbol{\beta}}_{s^c, \alpha}) + o(1/n_v) \\
&=  \frac{1}{n_v} \bigl(\boldsymbol{y}_s - \dot b(X_s^{\alpha}\tilde{\boldsymbol{\beta}}_{\alpha})\bigr)^{\top}X_s^{\alpha}\bigl(\dot u_{s^c}(\tilde{\boldsymbol{\beta}}_{\alpha})\bigr)^{-1}u_{s^c}(\tilde{\boldsymbol{\beta}}_{\alpha}) + o(1/n_v)\\
&+  \frac{1}{2 n_v} \bigl(\boldsymbol{y}_{s^c} - \dot b(X_{s^c}^{\alpha}\tilde{\boldsymbol{\beta}}_{\alpha})\bigr)^{\top}\bigl(\ddot b(X_s^{\alpha}\tilde{\boldsymbol{\beta}}_{\alpha})^{-1/2}\bigr)D_{s,\alpha}\bigl(D^{\top}_{s,\alpha}D_{s,\alpha}\bigr)^{-1} \\
&\times  \bigl((X_s^{\alpha})^{\top}\ddot b(X_s^{\alpha}\tilde{\boldsymbol{\beta}}_{\alpha})X_s^{\alpha}\bigr)\bigl((X_{s^c}^{\alpha})^{\top}\ddot b(X_s^{\alpha}\tilde{\boldsymbol{\beta}}_{\alpha})X_{s^c}^{\alpha}\bigr)^{-1} \\
&\times  D^{\top}_{s, \alpha}\bigl(\ddot b(X_s^{\alpha}\tilde{\boldsymbol{\beta}}_{\alpha})^{-1/2}\bigr)\bigl(\boldsymbol{y}_{s^c} - \dot b(X_{s^c}^{\alpha}\tilde{\boldsymbol{\beta}}_{\alpha})\bigr)(1+o(1)) \\
&=  B_{\alpha} + C_{\alpha}.
\end{align*}

By plugging in the expansion form of $\dot u_{s^c}(\cdot)$ and $u_{s^c}(\cdot)$,
\[
 B_{\alpha} = -\frac{1}{n_v}\bigl(\boldsymbol{y}_s - \dot b(X_s^{\alpha}\tilde{\boldsymbol{\beta}}_{\alpha})\bigr)^{\top}X_s^{\alpha}\bigl((X_{s^c}^{\alpha})^{\top}\ddot b(X_{s^c}^{\alpha}\tilde{\boldsymbol{\beta}}_{\alpha})X_{s^c}^{\alpha}\bigr)^{-1} (X_{s^c}^{\alpha})^{\top}\bigl(\boldsymbol{y}_{s^c} - \dot b (X_{s^c}^{\alpha}\tilde{\boldsymbol{\beta}}_{\alpha})\bigr)(1+o(1)).
\]
From Conditions~5 and 6, straight calculations lead to
\[
E(B_{\alpha}) = 0,\quad \mathrm{var}(B_{\alpha}) = d_{\alpha}a(\phi)(n_c n_v)^{-1/2}(1+o(1)).
\]
For $C_{\alpha}$ we have,
\begin{align*}
 C_{\alpha} &= \frac{1}{2 n_c} \bigl(\boldsymbol{y}_{s^c} - \dot b(X_{s^c}^{\alpha}\tilde{\boldsymbol{\beta}}_{\alpha})\bigr)^{\top}\bigl(\ddot b(X_s^{\alpha}\tilde{\boldsymbol{\beta}}_{\alpha})^{-1/2}\bigr) D_{s, \alpha}\bigl(D^{\top}_{s, \alpha}D_{s, \alpha}\bigr)^{-1}D^{\top}_{s,\alpha} \\
& \times \bigl(\ddot b(X_s^{\alpha}\tilde{\boldsymbol{\beta}}_{\alpha})^{-1/2}\bigr)\bigl(\boldsymbol{y}_{s^c} - \dot b(X_{s^c}^{\alpha}\tilde{\boldsymbol{\beta}}_{\alpha})\bigr)(1+o(1)).
\end{align*}
Thus, after taking expectation we have,
\[
E (A_{\alpha2, s}) = d_{\alpha}a(\phi)/n_c + o(1/n_c).
\]
If $\alpha \in \cA\setminus \cA_c$,
\begin{align*}
\tilde{\Gamma}_{\alpha_*} - \tilde{\Gamma}_{\alpha} = \frac{1}{n}\bigl(\ell_n(\tilde{\bbeta}_{\alpha_*}) - \ell_n(\tilde{\bbeta}_{\alpha})\bigr) + O(1/n_c).
\end{align*}
From Lemma~\ref{lemma-tala} and Condition~3, by exploiting Gaussian concentration, $\forall \varepsilon > 0$, we have
\begin{align*}
 R \cdot \mathrm{pr}\biggl \{n_c\biggl |\max_{\alpha \in \cA\setminus\cA_c} \Bigl|\frac{1}{n}\bigl(\ell_n(\tilde{\boldsymbol{\beta}}_{\alpha_*}) - \ell_n(\tilde{\boldsymbol{\beta}}_{\alpha})\bigr) \Bigr|
-  E\biggl(\max_{\alpha \in \cA\setminus\cA_c} \Bigl|\frac{1}{n}\bigl(\ell_n(\tilde{\boldsymbol{\beta}}_{\alpha_*}) - \ell_n(\tilde{\boldsymbol{\beta}}_{\alpha})\bigr) \Bigr|\biggr)\biggr| > \varepsilon \biggr\} \to 0.
\end{align*}
The parallel result for $\alpha \in \cA_c$ but $\alpha \neq \alpha_*$ holds similarly. Therefore, as $n \to \infty$, $\mathrm{pr} \{{\hat \alpha} \in \alpha_* \} \to 1$.

\end{proof}

\begin{table}[!h]
\caption{Comparisons in linear regression with $\rho=-0.5$. Results are reported in the form of mean (standard error).  FP, false positive; FN, false negative;  PE, prediction error.}
\label{tb::rho=-0.5}\par
\vskip .2cm
\begin{center}
\begin{tabular}{lrrr}
\hline
Method & \multicolumn{3}{c}{$\rho = -0.5$}\\ \hline
Lasso &  FP & FN & PE  \\ \hline
CV($n_v$)&0.03(0.02)&0.02(0.01)&1.01(0.01)\\
K-fold&30.53(2.84)&0.00(0.00)&1.09(0.01)\\
1SE&1.54(0.21)&0.00(0.00)&1.15(0.01)\\
AIC&469.97(1.39)&0.00(0.00)&1.38(0.01)\\
BIC&2.18(0.17)&0.00(0.00)&1.12(0.01)\\
EBIC&0.91(0.10)&0.00(0.00)&1.13(0.01)\\\hline
SCAD &  FP & FN & PE  \\ \hline
CV($n_v$)&0.06(0.03)&0.01(0.01)&1.01(0.01)\\
K-fold&24.48(2.70)&0.00(0.00)&1.03(0.01)\\
1SE&0.30(0.09)&0.00(0.00)&1.08(0.01)\\
AIC&25.20(2.02)&0.05(0.03)&1.09(0.03)\\
BIC&0.70(0.09)&0.05(0.03)&1.10(0.03)\\
EBIC&0.16(0.04)&0.05(0.03)&1.11(0.03)\\\hline
MCP &  FP & FN & PE  \\ \hline
CV($n_v$)&0.02(0.01)&0.00(0.00)&1.01(0.01)\\
K-fold&4.76(0.82)&0.00(0.00)&1.02(0.01)\\
1SE&0.04(0.04)&0.00(0.00)&1.07(0.01)\\
AIC&77.29(0.96)&0.00(0.00)&1.15(0.01)\\
BIC&0.52(0.11)&0.00(0.00)&1.02(0.01)\\
EBIC&0.06(0.03)&0.00(0.00)&1.02(0.01)\\
\hline
\end{tabular}
\end{center}
\end{table}

\noindent {\bf A.2 Additional numerical results}

We conducted an additional simulation  for the setting in Example 1(i) when $\rho=-0.5$ with the results summarized in Table \ref{tb::rho=-0.5}.  In this case, CV($n_v$) works very well compared with other methods and we skip the detailed  discussion since the message is very similar to the cases of $\rho=0$ and $\rho=0.5$.

%
%
\par



\end{document}